%% file: main.tex
\documentclass[10pt]{article}

\usepackage[margin=1in]{geometry} 
\usepackage{graphicx}
\usepackage[caption=false]{subfig}
\graphicspath{
    {figures/}
}

\usepackage{mathtools} 

\usepackage[english]{babel}
\usepackage{amsmath,amsthm,amssymb,amsfonts}
\usepackage[colorlinks=true, pdfstartview=FitV, linkcolor=blue, citecolor=blue, urlcolor=blue]{hyperref}
\usepackage[shortlabels]{enumitem}
\usepackage{bm}
\usepackage{dsfont} 
\usepackage{etoolbox}

\docsvlist{A,B,C,D,E,F,G,H,I,J,K,L,M,N,O,P,Q,R,S,T,U,V,W,X,Y,Z}

\docsvlist{A,B,C,D,E,F,G,H,I,J,K,L,M,N,O,P,Q,R,S,T,U,V,W,X,Y,Z}

\docsvlist{E,P,R,e,k,p,q,r,g,x,y}

\def\ve{h}
\def\bra{\langle}
\def\ket{\rangle}

\def\backslash{\symbol{92}}
\newcommand{\norm}[1]{\left\lVert#1\right\rVert}
\newcommand{\dnum}[1][1]{\mathds{#1}}

\DeclareMathOperator{\supp}{supp}

\newtheorem{theorem}{Theorem}[section]

\newtheorem{corollary}[theorem]{Corollary}
\newtheorem{prop}[theorem]{Proposition}
\newtheorem{lemma}[theorem]{Lemma}
\newtheorem{OQ}{Open Question}

\theoremstyle{definition}
\newtheorem{definition}[theorem]{Definition}

\newtheorem{remark}{Remark}

\usepackage{titlesec}

\newcommand{\addperiod}[1]{#1.}
\titleformat{\section}
   {\centering\normalfont\Large}{\thesection.}{0.5em}{}
\titleformat*{\subsection}{\bfseries}
\titleformat{\subsubsection}[runin]
  {\normalfont\bfseries}
  {\thesubsubsection.}
  {0.5em}
  {\addperiod}
\titleformat*{\subsubsection}{\bfseries}
\titleformat*{\paragraph}{\bfseries}
\titleformat*{\subparagraph}{\large\bfseries}

\newcommand{\Z}{\mathbb{Z}} 
 
\newcommand{\beq}{\begin{equation}}
\newcommand{\eeq}{\end{equation}}

\usepackage{xcolor}

\title{Infinite Stability in Disordered Systems}

\author{Andrew C. Yuan
\thanks{Department of Physics, Stanford University, Stanford, CA 94305, USA}
\thanks{Condensed Matter Theory Center and Joint Quantum Institute, Department of Physics, University of Maryland, College Park, Maryland 20742, USA;
{\footnotesize andrewcyuan@gmail.com}
}
\and 
Nick Crawford
\thanks{Department of Mathematics, The Technion, Haifa, Israel; \footnotesize nickc@technion.ac.il}
\date{} 
 }

\begin{document}

\maketitle

\begin{abstract}
    In quenched disordered systems, the existence of ordering is generally believed to be only possible in the weak disorder regime (disregarding models of spin-glass type). 
    In particular, sufficiently large random fields is expected to prohibit any finite temperature ordering. 
    Here, we prove that this is not necessarily true, and show rigorously that for physically relevant systems in $\mathbb{Z}^d$ with $d\ge 3$, disorder can induce ordering that is \textit{infinitely stable}, in the sense that (1) there exists ordering at arbitrarily large disorder strength and (2) the transition temperature is asymptotically nonzero in the limit of infinite disorder.
    Analogous results can hold in 2 dimensions provided that the underlying graph is non-planar (e.g., $\mathbb{Z}^2$ sites with nearest and next-nearest neighbor interactions).
\end{abstract}

\input{intro.tex}
\input{proof.tex}

\medskip
\noindent
\textbf{Acknowledgments.}
We thank Steven A. Kivelson, Julian May-Mann, Akshat Pandey, Christophe Garban, and Ron Peled for helpful discussions.
ACY was supported by NSF grant No. DMR-2310312 at Stanford and the Laboratory for Physical Sciences at CMTC.  NC was supported by ISF grant No. 1557-21.

\appendix

\input{app.tex}

\bibliographystyle{alpha}
\bibliography{main.bbl}

\end{document}

%% file: intro.tex
\section{Introduction}
\subsection{Classical XY Model}
In this paper, we consider various disordered versions of the XY model defined on finite sub-graphs $G$ of increasing volume in  a fixed infinite graph $G_{\infty}$.  Mostly $G_{\infty}$ will be $\Z^d$ with either a nearest neighbor finite range edge set a bilayer version $\Z^d\times\{0, 1\}$.

We begin by introducing the the classical XY model at inverse temperature $\beta \ge 0,$ gradually adding features to the model until we arrive at the main example studied here.  On $G$, the classical ferromagnetic XY model is an equilibrium spin system given by a Gibbs measure $\mu^\mathrm{XY}_{G,\beta},$ that is a probability measure $\mu^\mathrm{XY}_{G,\beta}$ on the set of spin configurations $\sigma \in (\dS^1)^V$.  It is convenient to parametrize $\sigma_x\in  \dS^1$ by the angle $\phi_x \in (-\pi,\pi]$.  With respect to these angular variables, the Gibbs state takes the form
\begin{equation}
    \label{eq:Gibbs-XY}
    \mu^\mathrm{XY}_{G,\beta}[d\phi] = \frac{1}{\sZ^\mathrm{XY}_{G,\beta}} \exp\left(-\beta \sH^\mathrm{XY}_G (\phi)\right) \prod_{x\in V}d\phi_x,
\end{equation}
Here $d\phi_i$ denotes the uniform measure on the interval $(-\pi,\pi]$, $\sZ^\mathrm{XY}_{G,\beta}$ is the partition function which normalizes the probability measure, and $\sH^\mathrm{XY}_G$ is the Hamiltonian of the system.
In this case given by
\begin{equation}
    \label{eq:H-XY}
    \sH^\mathrm{XY}_G (\phi) = -\sum_{e\in E} \cos \nabla_{\be} \phi.
\end{equation}
Here $\nabla_{\be} \phi = \phi_x-\phi_y$ denoting the phase gradient along the directed edge $\be=xy$, where an arbitrary choice of orientation is made at the outset.
Let us denote expectations with respect to $\mu^\mathrm{XY}_{G,\beta}$ by $\bra \cdot \ket_{G, \beta}^{XY}.$

Thermodynamic limits are obtained by considering possible weak limits of the family $(\mu^\mathrm{XY}_{G,\beta})_G$ as $G \nearrow G_\infty.$ The Ginibre inequalities \cite{ginibre1970general,tokushige2024graphical} imply that, as defined above with free boundary conditions, correlation functions described by cosines of sums of angular variables are increasing in $G.$ As a consequence the thermodynamic limit is unique in this case.  
In particular, for any vertices $x,y\in V_\infty$,
\begin{equation}
    \bra \cos(\phi_x -\phi_y)\ket_{G,\beta}^\mathrm{XY}\nearrow \bra\cos (\phi_x -\phi_y) \ket_{G_\infty,\beta}^\mathrm{XY},  \quad G\nearrow G_\infty
\end{equation}
Similarly, below we will fix a boundary condition in which all spins point in, say, the $X$ (horizontal) direction.  
In this case the Ginibre inequalities imply cosines of sums of angular variables are decreasing in $G$, and again the thermodynamic limit exists and is unique.

When restricted to $G_\infty =\dZ^d,d\ge 3$, the XY model is known to undergo a phase transition as one lowers the temperature $T$ (increases the inverse temperature $\beta\equiv 1/T$), i.e.,
\begin{itemize}
    \item At high $T\gg 1$ (where 1 is the scale of the coupling constants of the XY model in Eq. \eqref{eq:H-XY}), the model exhibits exponentially decaying correlations $\bra \cos(\phi_x -\phi_y)\ket_{\dZ^d,\beta}^\mathrm{XY}$ as $\norm{x-y}\to \infty$ \cite{aizenman1980comparison}.
    \item At low $T\ll 1$, the correlations exhibit \textit{long-range-ordering} (LRO), i.e., $\bra \cos(\phi_x -\phi_y) \ket_{\dZ^d,\beta}^\mathrm{XY} \ge c >0$ is bounded below and away from 0 as $\norm{x-y}\to \infty$ \cite{frohlich1976infrared,garban2022continuous}.
\end{itemize} 
Despite the lack of a rigorous proof, the dichotomy in long range behavior and the abundance of numerical evidence suggests that there exists a sharp phase transition and thus a well-defined transition temperature $T_\mathrm{LRO}^\mathrm{XY} \sim 1$ for $d\ge 3$ separating the high and low temperature phases.

The existence of a phase transition corresponds to the underlying global $U(1)$ symmetry in the model, invariance of the Hamiltonian and also of finite volume measures under uniform rotation of spins,  $\theta_x \mapsto \theta_x +\mathrm{const}$ for all sites $x$.
In particular, the model possesses $U(1)$ degenerate ground states (GS), i.e., any aligned spin configuration ($\theta_x = \mathrm{const}$)  minimizes the Hamiltonian.
 LRO at $T\ll 1$ implies nonergodicity of the infinite-volume system and suggests that the $U(1)$ symmetry has been ``spontaneously broken."

A more direct way to understand this symmetry breaking is to introduce boundary perturbations to the Hamiltonian in the process of taking thermodynamic limits.  
For definiteness, we consider the finite subgraphs induced by boxes $\Lambda_N =\{-N,...,N\}^d$.
Instead of $\mu_{\Lambda_N,\beta}^\mathrm{XY}$, one may be interested in understanding the conditional expectations 
\begin{equation}
    \mu^\mathrm{XY}_{0,\Lambda_N,\beta}[d\phi] \equiv \mu_{\Lambda_N,\beta}^\mathrm{XY}[d\phi | \phi_x=0,\forall x\in \partial \Lambda_N]
\end{equation}
where the XY spins $\phi$ are fixed on the boundary $\partial \Lambda_N \equiv \Lambda_N \backslash \Lambda_{N-1}$.
Due to the fixed boundary conditions, $\mu^\mathrm{XY}_{0,\Lambda_N,\beta}$ is not invariant under the $U(1)$ symmetry. 
Nevertheless, in the thermodynamic limit this setup leads to the same free energy as in the previous case.  

By Ginibre, the thermodynamic limit $\mu^\mathrm{XY}_{0,\dZ^d,\beta}$\cite{ginibre1970general,tokushige2024graphical} exists and in general
\begin{equation}
    \bra e^{i\phi_0}\ket^\mathrm{XY}_{0,\Lambda_N,\beta} \searrow \bra e^{i\phi_0}\ket^\mathrm{XY}_{0,\dZ^d,\beta} \geq 0, \quad N\nearrow \infty
\end{equation}
For $d\ge 3$,
\begin{itemize}
    \item At high $T\gg 1$, the model preserves the $U(1)$ symmetry in the sense that the local order parameter vanishes, i.e., $\bra e^{i\phi_0}\ket_{0, \dZ^d,\beta}^\mathrm{XY}=0$ \cite{aizenman1980comparison}.
    \item At low $T\ll 1$, the models exhibits \textit{spontaneous symmetry breaking} (SSB), i.e.,  $\bra e^{i\phi_0}\ket_{0, \dZ^d,\beta}^\mathrm{XY}>0$ \cite{frohlich1976infrared}.
\end{itemize} 
Despite the lack of a rigorous proof, it is expected that the XY model with fixed boundary conditions exhibits a sharp phase transition characterized by a transition temperature $T_\mathrm{SSB}^\mathrm{XY}$ which coincides exactly with that defined via free boundary conditions, i.e., $T_\mathrm{SSB}^\mathrm{XY} = T_\mathrm{LRO}^\mathrm{XY}$.
Similar symmetry considerations also occurs in discrete symmetry systems such as the Ising and Potts models, where the mathematical understanding is more developed.
In particular, sharpness of phase transitions and equivalence between LRO and SSB can be proven rigorously for such models.
See, for example, Ref. \cite{duminil2017lectures,friedli2017statistical} for a review.

In $d=2$ dimensions, the situation is more complicated. 
Due to the continuous $U(1)$ symmetry, the Mermin-Wagner theorem prohibits LRO/SSB at any finite temperature \cite{mermin1966absence,pfister1981symmetry,bricmont1977uniqueness}.
However, the XY model still exhibits a phase transition known as the Berezinskii-Kosterlitz-Thouless (BKT) transition \cite{berezinskii1971destruction,kosterlitz2018ordering,frohlich1981kosterlitz}. At $T\ll 1$, the correlations exhibit \textit{quasi}-LRO, i.e., algebraically decay with respect to $\norm{x-y}\to \infty$.
That there is a sharp transition from exponential to power law decay of correlations  was recently proved, \cite{van2023elementary,aizenman2021depinning}, and thus defines a critical transition temperature $T_\mathrm{QLRO}^\mathrm{XY}\sim 1$.
\subsection{Quenched Disorder}
We next consider the effects of disorder on the ferromagnetic XY model.
Viewed through of the lens of LRO, randomness in a system is usually regarded as an undesirable effect which disrupts coherence and thus lowers the transition temperature.
The Imry-Ma phenomenon \cite{imry1975random} provides an extreme example of this effect; in the presence of local random field interactions, the phase transition of spin systems can be completely destroyed in low dimensions independent of the field strength.

In the case of XY models, Imry-Ma applies to the following setup.
Let $\alpha:\dZ^d \to (-\pi,\pi]\cong \dS^1$ denote a collection of i.i.d. uniformly distributed random variables, and for each disorder realization $\alpha$, consider the finite-volume Gibbs measures $\mu^\mathrm{RFXY}_{G,\beta,\ve,\alpha}$ defined similarly as in Eq. \eqref{eq:Gibbs-XY} by the finite subgraph Hamiltonian
\begin{equation}
    \label{eq:H-RFXY}
    \sH^\mathrm{RFXY}_{G,\alpha,\ve} (\phi) = \sH^\mathrm{XY}_G(\phi) - \ve \sum_{x\in V} \cos(\phi_x-\alpha_x)
\end{equation}
If the disorder strength $\ve \ne 0$, then the RFXY model $\sH^\mathrm{RFXY}$ does not exhibit LRO/SSB at any finite temperature on $\dZ^d$ for $d\le 4$ dimensions.
This statement has been rigorously demonstrated for the random field XY (RFXY) model \cite{aizenman1989rounding,aizenman1990rounding}.  

The Imry-Ma phenomenon also applies to the random field Ising model \cite{aizenman1989rounding,aizenman1990rounding} and their quantum analogues \cite{greenblatt2009rounding,aizenman2012proof}.
More recent work even provides quantitative bounds on the decay of correlations in ground states \cite{ding2021exponential,aizenman2020exponential}.
Even in higher dimensions where that Imry-Ma is circumvented, there generally exists a critical threshold $\ve_c$ in disorder strength, above which any finite temperature ordering is destroyed \cite{aharony1978tricritical,aharony1978spin,bray1985scaling} (as illustrated in Fig. \ref{fig:schematic-phase}).
It may be worth mentioning that, for the RFXY model in dimension $d\leq 4$, the possibility of quasi-LRO has not been rigorously excluded {(partial results can be found in Ref. \cite{dario2024quantitative}).
Even within the physics community, there is no consensus on the behavior, where non-rigorous arguments and numerics both for \cite{feldman2001quasi,gingras1996topological} and against \cite{aharony1980infinite,tissier2006unified} the possibility of quasi-LRO have been proposed.}

The Imry-Ma argument applies when the random field is uniformly distributed over a compact spin space.  One may wonder, then, what is the fate of the system if the support of the disorder is restricted to a submanifold? 
In these models, even in low dimensions -- when Imry-Ma prohibits ordering along the submanifold -- the disorder actually generates ordering in directions perpendicular to the submanifold.

For example, consider the variant of the RFXY model in which  the random field disorder only  acts in the $X$ direction.  
This corresponds in Eq. \eqref{eq:H-RFXY} to assuming the $\alpha_x$ are  i.i.d. Bernoulli(1/2) with values in $=0,\pi$.
In this case, spins order in the $Y$ direction at finite temperatures, \textit{provided that the disorder strength is sufficiently weak}, $\ve \ll 1$, see \cite{dotsenko19812d,dotsenko1982spin,crawford2024random,crawford2013random,crawford2014random}. 
What is relevant to the present discussion is that in these models, it is still the case that at larger disorder strengths LRO is destroyed in all dimensions, \cite{aharony1978spin}, as the local random field determines the typical spin configurations.

In a somewhat different direction, another class of examples of interest are the XY models with random \textit{bond} disorder.  The most famous of these examples allows symmetric bond disorder on all edges and leads to spin glass behavior, but one might also consider breaking the symmetry of the disorder by allowing a bias in its distribution.  A particularly convenient special case of this is the selection of disorder on the Nishimori line \cite{Nishimori81,GHLDB85} (see e.g., Ref. \cite{ozeki1993phase} for a phase diagram of the 2D case).  The bias in these models is strongly correlated with $\beta,$ so that the larger $\beta$ is, the more ferromagnetic the edge couplings typically are.  In this sense, when $\beta$ is large in these models, the disorder effects are necessarily \textit{weak}.  In \cite{garban2022continuous}, it is demonstrated that ferromagnetic ordering persists on the Nishimori line for $\beta$ sufficiently large.

In light of the previous remarks, and leaving aside the subtle and complicated example of spin glasses, a natural question is whether there are spin systems which accommodate 
disorder that is {extensive} and {uncorrelated} over vertices of the underlying graph, but which nevertheless exhibit symmetry breaking at fixed values of $\beta$ \textit{over the entire range of disorder strengths $h$}.  We refer to this phenomenon as \textit{infinite stability} of LRO with respect to disorder strength.
In this article, we shall introduce a family of physically motivated examples  exhibiting this behavior.   We argue these systems exhibit LRO below a finite transition temperature \textit{independent} of the disorder strength both the weak \textit{and} strong disorder regimes (as illustrated in Fig. \ref{fig:schematic-phase}).

The spin configurations for our examples will be elements of $(\dS^1\times \dS^1)^{V_{\infty}}$, which one may also view as a system of $XY$ spins on the double layer graph $G_{\infty}\times \{+1, -1\}$.  Let $\alpha:V_\infty \to \{0,\pi\}$ be i.i.d. Bernoulli distribution with probability $p,1-p$ equal to $=0,\pi$ respectively.  Let $\dE_p$ denote the corresponding expectation.
Parametrizing spins by the angular variables $(\theta_x^+, \theta_x^-)\in (-\pi, \pi]^2$ for convenience,   denote the phase difference $\phi_x= \theta_x^+-\theta_x^-.$ For a fixed realization of disorder $(\alpha_x)$ consider the finite-volume  Hamiltonian
\begin{equation}
    \label{eq:H-alpha}
    \sH_{G,\alpha,\ve}(\theta^\pm) = \sum_{\ell=\pm} \sH^\mathrm{XY}_G (\theta^\ell) -\ve \sum_{x\in V} \cos (\phi_x -\alpha_x)
\end{equation}
The first summand indicates that each layer of spins possessing a conventional ferromagnetic $XY$ interaction. In addition, spins interact interlayer 
via a random phase shift.  Another way of saying this is that the interlayer interaction is a disordered ferromagnetic/antiferromagnetic one with signs chosen i.i.d. Bernoulli(p).  Let $\langle \cdot \rangle_{\Lambda_N}=\langle \cdot \rangle_{\Lambda_N, \alpha, h, \beta}$ denote the corresponding a.s. Gibbs state.

The question of interest is then whether the finite-volume disorder-averaged Gibbs measure $\dE_p \bra\cdot \ket_{\Lambda_N}$ induces a thermodynamic limit which exhibits LRO  in the strong disorder regime $\ve \gg 1$ (especially) below a transition temperature independent of the disorder strength.  The weak disorder regime is also of interest, but as we will argue, it should order according to a similar mechanism as in  the RFO(2) model.
By comparing with a more tractable approximation introduced below, we claim this model is infinitely stable in dimension $d\ge 3$ provided that the site percolation critical threshold satisfies $p_c^\mathrm{site}(\dZ^d) < \min(p,1-p)$.  The precise statement is collected in our main result in Theorem \eqref{thm:infinite-stability}.
This article complements a shorter version \cite{yuan2024concise}, in which many technical details and discussions were omitted.

\begin{figure}[ht]
\centering
\includegraphics[width=.5\columnwidth]{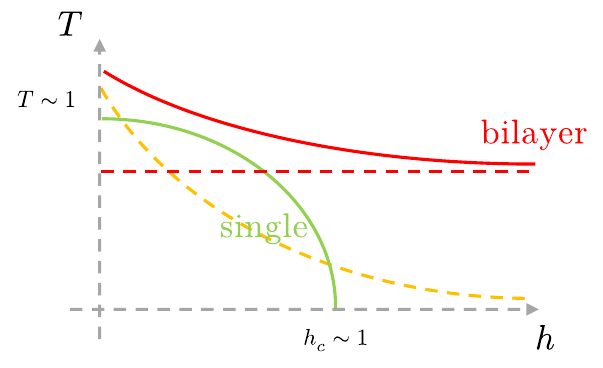}
\caption{Schematic Phase Diagram. The green line sketches the transition temperature $T_c(\ve)$ of \textit{single layer} models with disorder strength $\ve$, e.g., the random-field Ising  and XY models \cite{aharony1978tricritical,aharony1978spin,bray1985scaling}, the RFO(2) model \cite{dotsenko19812d,dotsenko1982spin,crawford2024random,crawford2013random,crawford2014random}.
If $\ve \to 0$, then $T_c(\ve)\sim 1$ is on the scale of the XY coupling strength $\kappa=1$ (in 2D, there is debate on whether LRO or QLRO occurs below $T_c$ \cite{crawford2024random}).
For these models, there exists a critical threshold $\ve_c \sim 1$, beyond which there is no finite temperature ordering.
In comparison, this manuscript introduces a \textit{bilayer} model with $T_c(\ve)$ sketched by the red line.
If $\ve \to \infty$, then $T_c(\ve)$ is asymptotically nonzero (dashed red line).
For completeness, the dashed orange line denotes an intermediate scenario where there is finite temperature ordering at large disorder, but $T_c(\ve)$ is asymptotically zero.
}
\label{fig:schematic-phase}
\end{figure}
\subsection{Main Results}
As given, the model defined by Eq. \eqref{eq:H-alpha} is difficult to fully control in either the weak or strong disorder regimes and we do not provide such an analysis here.

Nevertheless within the strong disorder regime, we will argue  heuristically in the next section that the model 
is well approximated by the \textit{effective} Hamiltonian $\sH^\mathrm{s}$
\begin{equation}
    \label{eq:H-strong}
    \sH_{G,\alpha,\ve}^\mathrm{s}(\theta) = 
    -2 \sum_{e\in E:\nabla_{\be} \alpha =0} \cos \nabla_{\be} \theta  -\frac{1}{2\ve}\sum_{x\in V} (\nabla_x^\alpha \cdot \cos \nabla\theta)^2,
\end{equation}
in the strong disorder regime $h\gg 1\gg T$ (see Prop. \eqref{prop:motivation}) below.
Here the ``divergence" term involving $\nabla_x^\alpha \cdot$ is defined by
\begin{equation}
    x\mapsto \nabla_x^\alpha \cdot \cos \nabla \theta  \equiv \sum_{e\sim x: \nabla_{\be} \alpha \ne 0} \cos \left(\nabla_{\be} \theta\right), 
\end{equation}
where the summation is over all edges $e \in E$ incident to $x$.  The main advantage of the Hamiltonian $\sH^{\mathrm{s}}$ over the original model is that it satisfies the conditions to apply the Ginibre inequalities, for any disorder realization.
Under the assumption that the phase difference $\phi_x$ should be close to the random phase shift $\alpha_x$,  the two models are connected by the relation $e^{i\phi} \approx e^{i\alpha} + i\tau,$ where 
\begin{equation}
\label{eq:order-parameter}
\tau = \tau^\alpha (\theta) = \frac{1}{\ve}\nabla^\alpha \cdot \cos\nabla\theta
\end{equation}
We then conjecture that 
\begin{align}
    \label{eq:heuristic}
    \bra e^{i (\phi_x-\phi_y)} \ket_{G,\alpha,\ve,\beta} &\approx  e^{i (\alpha_x-\alpha_y)} +\left\bra \tau_x \tau_y \right\ket_{G,\alpha,\ve,\beta}^\mathrm{s}\\
    \dE_p \bra e^{i (\phi_x-\phi_y)} \ket_{G,\alpha,\ve,\beta} &\approx   \dE_p \left\bra \tau_x \tau_y \right\ket_{G,\alpha,\ve,\beta}^\mathrm{s}.
\end{align}
On the right-hand-side (RHS) is the Gibbs state corresponding to $\sH^\mathrm{s}$ on the finite subgraph $G=(V,E)$.

Let $p_c = p_c^\mathrm{site}(\dZ^d)$ denote the critical point for site percolation on $\Z^d$.
It is known, see \cite{CampaninoRusso85}, that $p_c<1/2$ for $d\geq 3$ and hence, for $p\in (p_c, 1-p_c)$, connected components of  $\{x: \alpha_x=0\}$ and $\{x: \alpha_x=\pi\}$ both percolate.

\begin{theorem}[Infinite Stability]
    \label{thm:infinite-stability}
    Consider the effective Hamiltonian $\sH^\mathrm{s}$ in Eq. \eqref{eq:H-strong} on $\dZ^d,d\ge 3$.  Suppose $\min(p,1-p) > p_c^\mathrm{site}(\dZ^d).$  Then, there are $\ve_0, \beta_0=\beta_0(d, p)>0$ so that for any $\ve \ge \ve_0 >0$ and  $\beta \ge \beta_0$, there is a $c>0$ depending on $h, \beta, p$ so that
    \begin{equation}
        \dE_p \left[  \bra \tau_x \tau_y\ket^\mathrm{s}_{\dZ^d,\alpha,\ve,\beta}\right]\geq c>0.
    \end{equation}
\end{theorem}

In $d=2$ dimensions, the situation is bit more complicated. Since $\min(p,1-p)$ obtains its maximum value when $p=1/2$ and $p_c^\mathrm{site}(\dZ^2) \ge 1/2,$ the condition of the previous theorem can never be satisfied.
In fact, any infinite planar graph satisfies $p_c^\mathrm{site}(G_\infty) \ge 1/2$ \cite{grimmett2022hyperbolic}, and thus it's likely that the effective model, with Hamiltonian $\sH^\mathrm{s}$ can, at best, exhibit LRO with a transition temperature that decays asymptotically to zero (as illustrated by the orange dashed line in Fig. \ref{fig:schematic-phase}).  

The situation looks better if the underlying graph is non-planar.
For example, if we consider the non-planar graph $\tilde{\dZ}^2$ with vertices $\dZ^2$ and edges consisting of nearest and next-nearest 
neighbors, i.e., $e=xy$ such that $\norm{x-y}_\infty=1$, then $p_c^\mathrm{site}(\tilde{\dZ}^2)< 1/2$ \cite{malarz2005square} (This also follows by comparing with site percolation on the triangular lattice obtained by projecting in the $(1,1)$ direction \cite{CampaninoRusso85}.). 
In particular, we can show that the system is at least \textit{infinitely quasi-stable}, i.e.,

\begin{theorem}[Infinite Quasi-Stability]
    \label{thm:infinite-quasi-stability}
   Consider the effective Hamiltonian $\sH^\mathrm{s}$ in Eq. \eqref{eq:H-strong} on $\tilde{\dZ}^2$.  Suppose $\min(p,1-p) > p_c^\mathrm{site}(\tilde{\dZ}^2).$  Then, there are $\ve_0, \beta_0=\beta_0(d, p)>0$ so that for any $\ve \ge \ve_0 >0$ and  $\beta \ge \beta_0$, there is are constants $C, c>0$ depending on $h, \beta, p$ so that
    \begin{equation}
        \dE_p \left[  \bra \tau_x \tau_y\ket^\mathrm{s}_{\dZ^d,\alpha,\ve,\beta}\right]\geq C\|x-y\|^{-c}_{\infty}>0.
    \end{equation}
\end{theorem}
\subsection{Background}
Physically, the XY model is known to be the classical description of superconductors (SC) so that its phase transition is interpreted as the normal metal to SC transition as one lowers the temperature $T$ \cite{benfatto2004low,roddick1995effect,emery1995superconductivity,carlson1999classical}.
Recently, there has been an increasing interest in the physics community on twisted bilayer systems, both experimentally \cite{cao2018correlated,cao2018unconventional,yankowitz2019tuning,zhao2023time} and theoretically \cite{can2021high,how2023absence,liu2023charge,yuan2023exactly,yuan2023inhomogeneity,yuan2024absence,song2022phase} (see Ref. \cite{pixley2025twisted} for a recent review).
In these systems, each layer is a 2D superconductor and thus described by an independent 2D XY models $\sH^\mathrm{XY}(\theta^\ell)$ for layers $\ell =\pm$. 
The dominant inter-layer interaction is generally described by the Josephson coupling that depends only on the phase difference $\phi \equiv \theta^+ -\theta^-$ and schematically written as an additional term $-J \cos \phi$ in the Hamiltonian, i.e., the finite subgraph Hamiltonians are written as
\begin{equation}
    \label{eq:H-J}
    \sH_{G,J} (\theta^\pm) = \sum_{\ell=\pm}\sH_{G}^\mathrm{XY}(\theta^\ell) -\sum_{x\in V} J_x \cos \phi_x, \quad J:V_\infty \to \dR.
\end{equation}
The inter-layer coupling constants $J:V_{\infty} \to \dR$ are treated as quenched disorder which are i.i.d. among vertices with mean coupling strength $\bar{J} \in \dR$ and standard deviation $\ve \ge 0$.
By twisting the relative orientation between the two layers, physicists can vary the mean Josephson coupling $\bar{J} \in \dR$. 
The disorder fluctuations $\delta J_x \equiv J_x -\bar{J}$ can then arise from twist angle disorder, i.e., the inability to control the relative twist at a microscopic level, or other inhomogeneities \cite{yuan2023inhomogeneity,zhao2023time}.
Since the Hamiltonian possesses a global $\dZ_2$ symmetry via $\phi \mapsto -\phi$, physical intuition suggests that the phase difference $\phi$ can exhibit a finite temperature phase transition provided that the GS is degenerate.

At most twist angles, the average coupling $\bar{J}$ is much larger than the disorder strength $\ve$, i.e., $|\bar{J}| \gg \ve$.
In this scenario, it's natural to ignore disorder fluctuations and treat the system as if it exhibits a uniform inter-layer interaction with coupling strength $\bar{J}$.
In this case, it's straightforward to check that the GS is uniquely described by the solution where the phase difference $\phi_x \equiv 0$  (or $\pi$) uniformly in vertices $x\in V_\infty$ provided that $\bar{J} >0$ (or $\bar{J} <0$).
Since the GS with respect to $\phi$ is unique, it's expected that there is no nontrivial phase transition at finite temperatures.
However, at certain critical twist angles, the Josephson coupling vanishes on average $\bar{J} =0$ due to underlying symmetries of the SC \cite{yuan2023inhomogeneity}.
Hence, near the critical twist angle so that $|\bar{J}| \ll \ve$, the disorder fluctuations $\delta J$ cannot be ignored and thus must be treated accordingly.

In comparison with the general Hamiltonian in Eq. \eqref{eq:H-J}, we shall be content with studying the restricted scenario where $J_x = \pm \ve$ is bimodal.
In this case, the restricted model in Eq. \eqref{eq:H-alpha} corresponds to the scenario with average coupling $\bar{J} =(2p-1)\ve$ and variance $4\ve^2 p(1-p)$.
In particular, if $\alpha$ is equally likely to be $0$ or $\pi$ so that $p=1/2$, then quenched disorder $\bar{J}=0$ vanishes on average but has disorder strength $\ve$.
Despite our physical motivation of twisted bilayers, the model is straightforwardly generalized to higher dimensions (in fact, on any infinite graph $G_\infty$), which may model random inter-component interaction within multi-component SCs \cite{carlstrom2011length,bojesen2014phase,grinenko2021state}.


\section{Heuristics for weak and strong disorder}
\subsection{Strong disorder.}
\subsubsection{Derivation of \eqref{eq:H-strong} from \eqref{eq:H-alpha}}
\label{sec:strong-disorder}

\begin{figure}[ht]
    \centering
    \subfloat[\label{fig:average-a}]{%
        \centering
        \includegraphics[width=0.3\columnwidth]{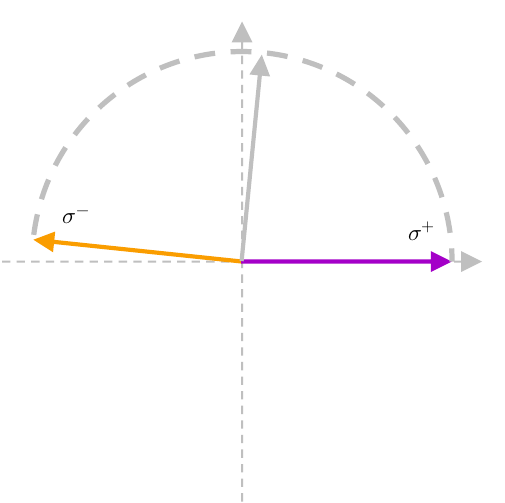}
    }
    \subfloat[\label{fig:average-b}]{%
        \centering
        \includegraphics[width=0.3\columnwidth]{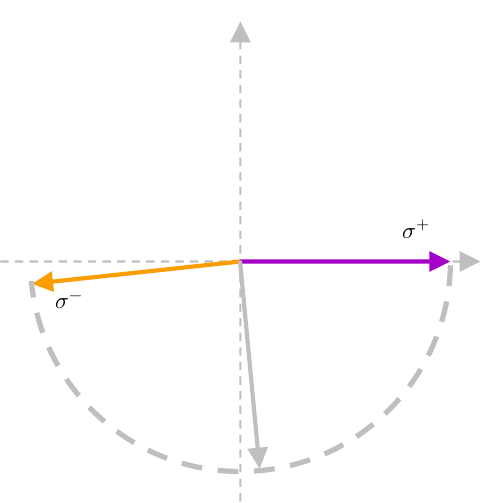}
    }
    \subfloat[\label{fig:measure-a}]{%
        \centering
        \includegraphics[width=0.3\columnwidth]{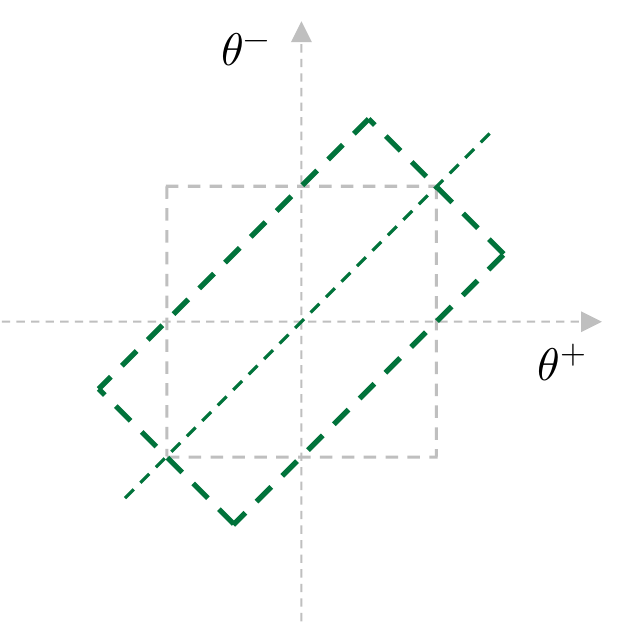}
    }
    \\
    \subfloat[\label{fig:average-c}]{%
        \centering
        \includegraphics[width=0.3\columnwidth]{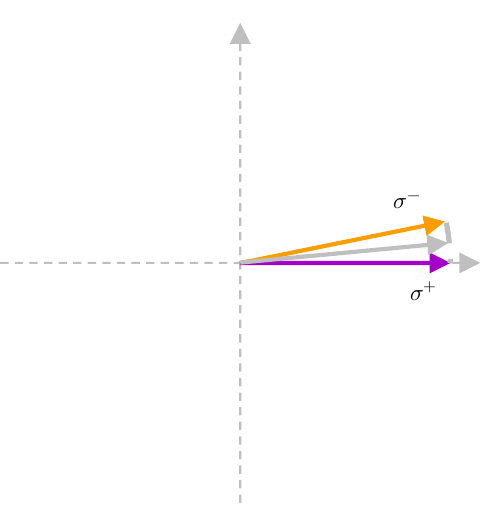}
    }
    \subfloat[\label{fig:average-d}]{%
        \centering
        \includegraphics[width=0.3\columnwidth]{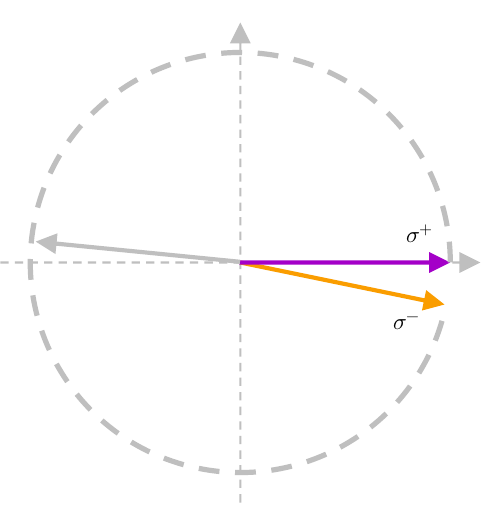}
    }
    \subfloat[\label{fig:measure-b}]{%
        \centering
        \includegraphics[width=0.3\columnwidth]{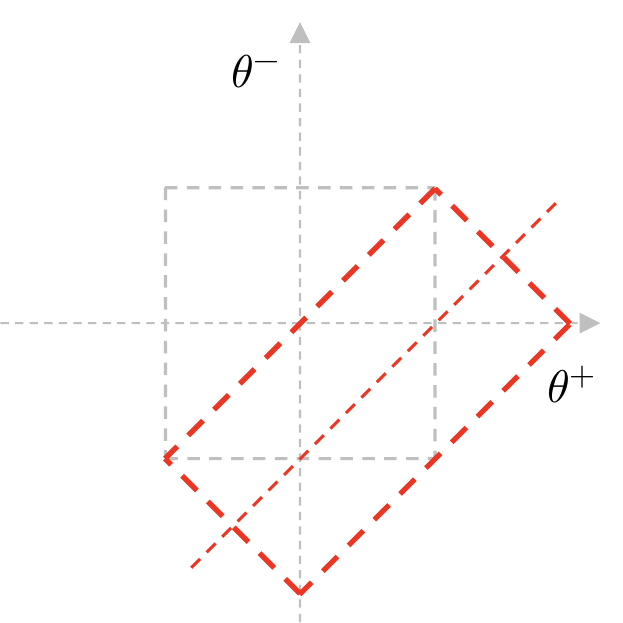}
    }
    \caption{The strong disorder change of variable. In each subfigure, the purple, orange arrows denote the $\sigma^\pm =e^{i\theta^\pm}$ spins, respectively. 
    Subplots (a), (b) denote the change of variables described by Eq. \eqref{eq:COV}, $\alpha_x=0$, where the grey arrow indicates $\zeta$ pointing in the middle of the smaller region. Note the discontinuity near $\sigma^+ =-\sigma^-$.
    Subplots (d), (e) denote that described by Eq. \eqref{eq:COV}, $\alpha_x=\pi$, where the grey arrow indicates $\zeta$ pointing in the middle of the region characterized by the counter-clockwise orientation from $\sigma^+$ to $\sigma^-$. Note the discontinuity near $\sigma^+ = \sigma^-$.
    The grey box in subplots (c), (f) denote the conventional $(-\pi,\pi)^2$ box. 
    If $\alpha_x=0$ and $\theta^\pm$ are restricted within the green box in subplot (c), then $\zeta = e^{i\theta},w=e^{i\phi}$ where $\theta = (\theta^++\theta^-)/2,\phi = \theta^+ -\theta^-$ and $(\theta,\phi)\in (-\pi,\pi)^2$.
    The green dashed line denotes $\phi =0$.
    Similarly, if $\alpha_x=\pi$ and $\theta^\pm$ are restricted within the red box in subplot (f), then $\zeta = e^{i\theta},w=e^{i(\phi-\pi)}$ where $\theta = (\theta^++\theta^-)/2,\phi = \theta^+ -\theta^-$ and $(\theta,\phi-\pi)\in (-\pi,\pi)^2$.
    The red dashed line denotes $\phi = \pi$.
    }
    \label{fig:change-of-variables}
\end{figure}
Starting from the Hamiltonian \eqref{eq:H-alpha},  in the strong disorder $\ve\gg 1$ regime, it's instructive to make a change of variables so that the phase difference $\phi$ is an independent variable.

Let us view $\sigma_x^{\pm}=e^{i\theta_x^\pm}$ as complex numbers with modulus $1$. We begin with the following elementary observation.
\begin{prop}
    \label{prop:motivation1}
  Given $\alpha_x\in \{0, \pi\}$  consider the map
    \begin{equation}
    \label{eq:COV}
        \zeta_x = \frac{e^{-i\alpha_x/2}\sigma^+_x +e^{i\alpha_x/2}\sigma^-_x}{|e^{-i\alpha_x/2}\sigma^+_x +e^{i\alpha_x/2}\sigma^-_x|},\quad
        w_x     = \sigma^+_x \bar{\sigma}^-_x e^{-i\alpha_x}.
    \end{equation}
    This map is smooth on 
    \begin{equation}
        (\dS^1 \times \dS^1 ) \backslash \{\sigma^+_x =-e^{i\alpha_x}\sigma^-_x\} \to \dS^1 \times \dS^1 \backslash\{-1\}
    \end{equation}
    Moreover, the mapping preserves the \textit{a priori} measure on $\dS^1 \times \dS^1$ up to a set of measure $0$. 
\end{prop}

\begin{proof}
Observe that, with respect to the uniform (probability) distribution on $\mathbb S^1\times \mathbb S^1,$ the change of variable $(\sigma^-_x, w_x)$ is measure preserving.  Note that since the distribution of $w$ is continuous, the variable
\begin{equation}
    \zeta_x=e^{i\alpha_x/2}\sigma^-_x\frac{1+w_x}{|1+w_x|},
\end{equation}
is defined a.s. It follows by independence of $w, \sigma^-$ that, conditional on $w_x,$  $\zeta_x$ is uniformly distributed over $\mathbb S^1.$  Therefore $(\zeta_x, w_x)$ is measure preserving and defined up to a set of measure $0.$
\end{proof}

Using this change of variable, we re-express $\sH_{G,\alpha,\ve}.$  Let us introduce the notation  $\zeta_x = e^{i\theta_x}$, $\tau_x=\Im[e^{i\alpha_x} w_x]$.  For a high density of vertices, under the Gibbs measure $|\tau_x|\ll 1$ and so, it makes sense to study the leading order behavior of $\sH_{G,\alpha,\ve}$ with respect to $\tau.$

Recall the divergence $\nabla_x^\alpha \cdot$ is naturally defined as
    \begin{equation}
      \nabla_x^\alpha\cdot \cos \nabla\theta = \sum_{e\sim x:\nabla_{\be} \alpha \ne 0} \cos \nabla_{\be} \theta
    \end{equation}

\begin{prop}
\label{prop:motivation}
Let $\sH_{G,\alpha,\ve}$ denote the bilayer Hamiltonian on finite graph $G=(V,E)$ given in Eq. \eqref{eq:H-alpha}. Up to additive constants,
\begin{align}
    \sH_{G,\alpha,\ve} &= -2\sum_{e\in E:\nabla_{\be} \alpha =0} \cos \nabla_{\be} \theta 
    -\frac{1}{2{\ve}}\sum_{x\in V} (\nabla_x^\alpha \cdot \cos \nabla\theta)^2\\
    &+\frac{{\ve}}{2} \sum_{x\in V} \left(\tau_x - \frac{1}{{\ve}} \nabla_x^\alpha\cdot \cos \nabla\theta\right)^2+\sum_{e=xy\in E}\Phi_{e}(\theta, \tau),
\end{align}
where  
\begin{equation}
|\Phi_e(\theta, \tau)|=O\left([|\tau_x|^2+ |\tau_y|^2]\;[| \cos\nabla_{\be}\theta|+ h(|\tau_x|^2+|\tau_y|^2)]\right).
\end{equation}

\end{prop}
\begin{proof}
    Let us write $w_x =a_x +ib_x.$ Assuming $|b_x|\ll 1$ in any reasonable sense (either pointwise, or perhaps in some $L^p$ space),
    the inter-layer interaction in Eq. \eqref{eq:H-alpha} becomes
    \begin{equation}
        -h \cos(\phi_x -\alpha_x) = -ha_x = -h\left(1-\frac{b_x^2}{2} +O(b_x^4)\right)
    \end{equation}
    In the new variables, we have
    \begin{equation}
        \sigma_x^+ = e^{+i\alpha/2} \frac{|1+w_x|}{1+\bar{w}_x}\zeta_x, \quad \sigma_x^- = e^{-i\alpha/2} \frac{|1+w_x|}{1+w_x}\zeta_x.
    \end{equation}
    Hence, 
    \begin{align}
        \frac{1}{2}\sum_{\ell=\pm} (\sigma_x^\ell \bar{\sigma}_y^\ell +\bar{\sigma}_x^\ell {\sigma}_y^\ell) = 2 \Re[\zeta_x \bar{\zeta}_y] \Re \left[e^{-i\nabla_{\be} \alpha/2} \frac{(1+\bar{w}_x)(1+w_y)}{|(1+\bar{w}_x)(1+w_y)|} \right].
    \end{align}
    Further note that
    \begin{align}
        \Re \left[e^{-i\nabla_{\be} \alpha/2}\frac{(1+\bar{w}_x)(1+w_y)}{|(1+\bar{w}_x)(1+w_y)|}\right]
        =
        \begin{dcases}
            \Re\left[\frac{(1+\bar{w}_x)(1+w_y)}{|(1+\bar{w}_x)(1+w_y)|} \right] & \nabla_{\be} \alpha = 0\\
            \pm \Im\left[\frac{(1+\bar{w}_x)(1+w_y)}{|(1+\bar{w}_x)(1+w_y)|} \right] & \nabla_{\be} \alpha =\pm \pi
        \end{dcases}
    \end{align}
    Using
    \begin{align}
        (1+\bar{w}_x)(1+w_y) &= (1+a_x)(1+a_y) +b_xb_y +i\left[(1+a_x)b_y+(1+a_y)b_x\right] \\
        &=4 + O(|b_x|^2+|b_y|^2)+i\left[2(b_y-b_x) +O(|b_x|^3+|b_y|^3)\right]
    \end{align}
    we get
    \begin{align}
        \Re \left[e^{-i\nabla_{\be} \alpha/2} \frac{(1+w_x)(1+\bar{w}_y)}{|1+w_x||1+w_y|} \right]
        =
        \begin{dcases}
            1+O(b_x^2+b_y^2) & \nabla_{\be} \alpha = 0\\
            \frac{1}{2} (e^{i\alpha_x} b_x+ e^{i\alpha_y}b_y) +O(|b_x|^3+|b_y|^3) & \nabla_{\be} \alpha =\pm \pi
        \end{dcases}
    \end{align}
    We can therefore write
    \begin{align}
        \sH_{G,\alpha,\ve} &= -2\sum_{e\in E:\nabla_{\be} \alpha =0} (1+O(|b_x|^2+|b_y|^2)) \cos \nabla_{\be} \theta \\ 
        &\quad-\sum_{x\in V} e^{i\alpha_x}b_x \nabla_x^\alpha \cdot \cos \nabla_{\be} \theta +\frac{h}{2} \sum_{x\in V} b_x^2  +O(h) \sum_{x\in V} b_x^4 
    \end{align}
  
The lemma follows by completing the square with respect to $\tau_x$.    
Note that a slightly more precise accounting would include the effect of the term $O(|b_x|^2+|b_y|^2)) \cos \nabla_{\be} \theta$ in the quadratic minimization for $\tau$.  However, when $h$ is large this effect is clearly negligible.  Since we cannot rigorously treat the full model for other reasons, we ignore this term in the present approximation.
\end{proof}

The point of Prop \eqref{prop:motivation} is that in the large disorder regime and conditional on $\zeta$ we should view the field $\tau_x,$ approximately,  as a family of independent normal random variables with
\begin{equation}
    \label{eq:phi-conditional-theta}
    \tau_x|_\zeta \sim \sN\left(\frac{1}{h} \nabla_x^\alpha\cdot \cos \nabla\theta, \frac{1}{\beta h}\right)
\end{equation}
This assumption is, at least, self-consistent since,  conditional on $\zeta$, this field has mean $O(1/h)\ll 1$ and variance $O(1/(\beta h))\ll 1$ provided  $h\gg 1 \gg T.$  In particular $\tau_x^4 \ll \tau_x^2$ with high probability. 
There are ways of arguing rigorously that $\tau_x$ is typically small, so the main difficulty in controlling the model rigorously concerns the symmetry breaking of $\zeta$.  We next discuss heuristics for the behavior of $(\zeta, \tau)$ with respect to $\sH^\mathrm{s}$ in Eq. \eqref{eq:H-strong}.



\begin{remark}
    In our shorter paper \cite{yuan2024concise}, we wrote our effective Hamiltonian in terms of the average phase $\theta \equiv (\theta^+ +\theta^-)/2$ and phase difference $\phi \equiv \theta ^+ -\theta^-$ to appeal to physical intuition, where both are restricted to the box $|\theta|,|\phi-\alpha| < \pi$. 
    However, since $\theta^\pm \mapsto \theta$ is not $2\pi$-invariant, it raises some doubts on the parametrization.
    Here, we have derived the effective Hamiltonian in more invariant manner and from our previous discussion, when we restrict our parametrization of $\sigma^\pm =e^{i\theta^\pm}$ to the green/red boxes in Fig. \ref{fig:measure-a}, \ref{fig:measure-b}, the new variables $\theta,\phi$ are indeed restricted to the box $|\theta|,|\phi-\alpha|<\pi$ and that $\zeta = e^{i\theta}, w=e^{i(\phi-\alpha)}$.
    More importantly, we note that since the imaginary part of $w$ is presumably small by Eq. \eqref{eq:phi-conditional-theta} in the regime $h\gg 1 \gg T$, the imaginary part $\Im w \approx \phi -\alpha$ (the latter was denoted $\delta \phi$ in Ref. \cite{yuan2024concise}).
\end{remark}

\begin{figure}[ht]
\centering
\subfloat[\label{fig:schematic-lattice}]{%
    \centering
    \includegraphics[width=0.33\columnwidth]{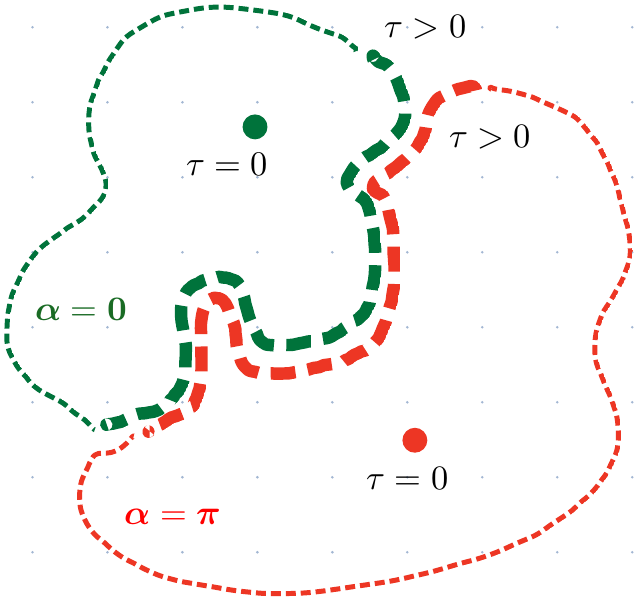}
}
\subfloat[\label{fig:schematic-spin}]{%
    \centering
    \includegraphics[width=0.33\columnwidth]{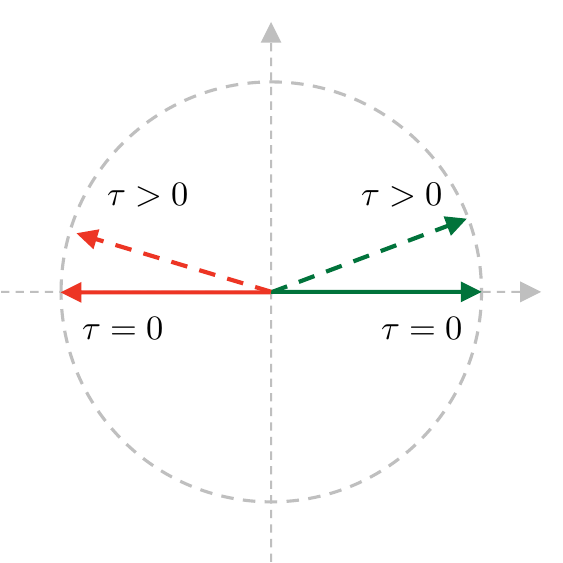}
}
\subfloat[\label{fig:GS-a}]{%
    \centering
    \includegraphics[width=0.33\columnwidth]{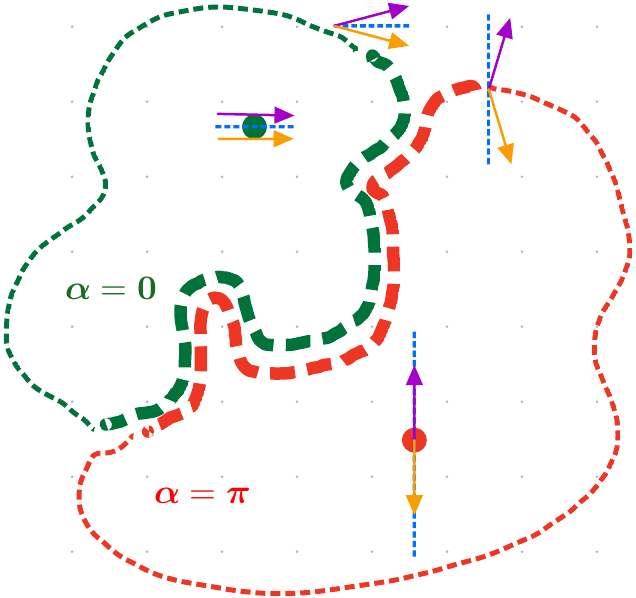}
}
\caption{GS Schematic $e^{i\phi} \sim^\text{s} +i$.  In the strong disorder regime, (a) shows a typical disorder realization of clusters with random phase $\alpha=0,\pi$ where the green and red dashed lines denote the cluster boundary. By Eq. \eqref{eq:GS} (with the assumption that $\nabla \theta =0$), the slanting $\tau_x=0$ vanishes within the interiors of the clusters (e.g., the colored dots deep within each cluster). 
However, at the boundary of each cluster, $\tau_x > 0$, since there is a nontrivial jump in neighboring phase gradients $\nabla \alpha = \pm \pi$. (b) shows the corresponding GS $\phi$ in spin space $\dS^1$, given by the heuristics $e^{i\phi} = e^{i\alpha} +i\tau$ in Eq. \eqref{eq:heuristic}. At the boundary (dashed lines), there is uniform slanting in the $+i$ direction.
(c) shows the GS in the original spins $\sigma^\pm$ denoted by the purple, orange arrows, respectively.
}
\label{fig:GS}
\end{figure}

\subsubsection{Ground States for $\sH^\mathrm{s}$ in Eq. \eqref{eq:H-strong}}
Let us partition the lattice sites into site-connected percolation clusters\footnote{Neighboring sites with the same random phase $\alpha$ are connected.} of $\alpha=0,\pi$ (see Fig. \ref{fig:GS}).  
As we have already argued, though the order parameter of interest is $\tau_x$, we only need to consider the influence of $\zeta$ on it.  
To determine the low energy behavior of $\zeta$  under $\sH^\mathrm{s}$ in Eq. \eqref{eq:H-strong}, we can ignore $\tau$.  
Then, to leading order it's reasonable to assume that the average phase $\theta$ is long-range-ordered.  
This then implies that for a typical low energy configuration, (with $\nabla \theta =0$),
\begin{equation}
    \label{eq:GS}
    \mathbb E[\tau_x| \zeta]  =   \frac{1}{\ve} |\{e\sim x: \nabla_{\be} \alpha \ne 0\}|
\end{equation}
From Eq. \eqref{eq:GS}, if $x$ is within the interior of a cluster, then the mean of $\tau_x = 0$ vanishes whereas if $x$ is at a cluster boundary, the mean is positive with magnitude $\sim 1/\ve \ll 1$.
In this sense, the ordering of the GS $\tau$ only occurs on $\alpha$-cluster boundaries, where $e^{i\phi} \approx e^{i\alpha}+i\tau$ (see Eq. \eqref{eq:heuristic}) uniformly slants in the $+i$ direction. (the alternative $\dZ_2$ degenerate GS $e^{i\phi} \sim^\text{s} -i$ is such that $e^{i\phi}$ slants in the $-i$ direction on the boundaries).
The crucial point is that, even though slanting is only present at the interface between clusters, it is still a bulk effect between a pair of \textit{infinite} components.  
If $p_c^{\text{site}} <1/2$, both the $\alpha=0,\pi$ sites percolate, each with a unique infinite cluster \cite{burton1989density} and thus ergodicity dictates that the interface between the two infinite clusters is volume extensive (see Fig. \ref{fig:cluster}).

\begin{figure}[ht]
\centering
\subfloat[\label{fig:tildeZ2}]{%
    \centering
    \includegraphics[width=0.3\columnwidth]{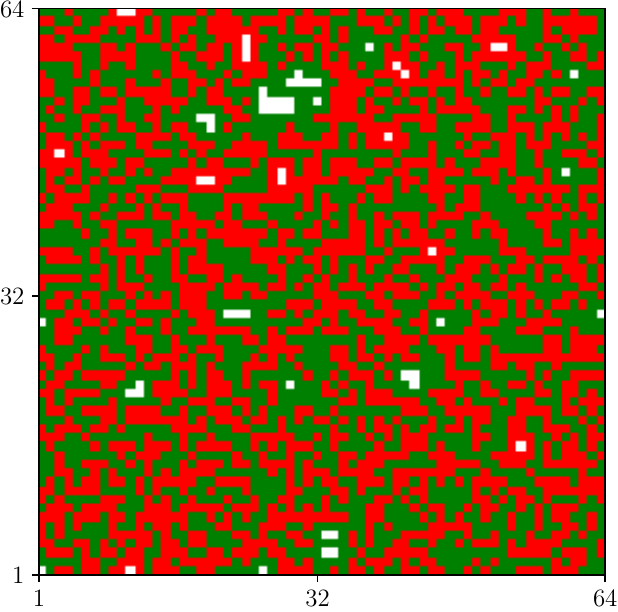}
}
\subfloat[\label{fig:Z2}]{%
    \centering
    \includegraphics[width=0.3\columnwidth]{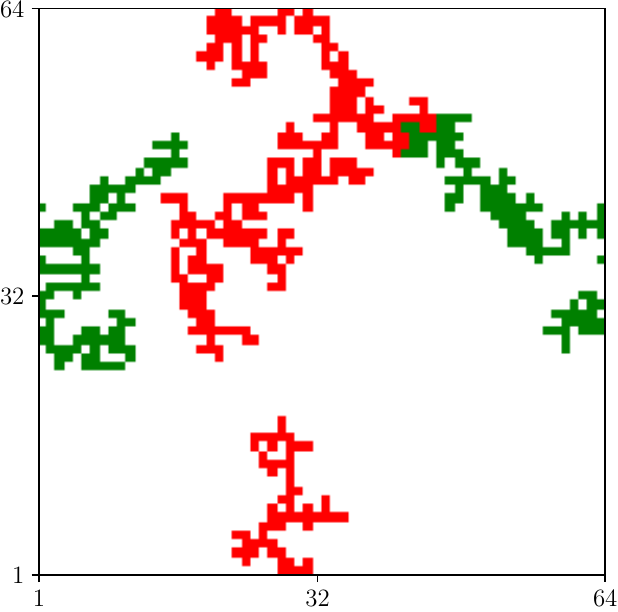}
}
\caption{Largest cluster.  For a given disorder realization $\alpha$ (with periodic boundary conditions), (a) plots the largest $\alpha=0$ cluster (green) and $\alpha =\pi$ cluster (red) on $\tilde{\dZ}^2$, while (b) plots those on $\dZ^2$. All other lattice sites are colored white. The difference between subplots (a), (b) is due to $p_c^\text{site}(\tilde{\dZ}^2) < 1/2 <p_c^\text{site} (\dZ^2)$
}
\label{fig:cluster}
\end{figure}
\subsubsection{Finite temperature}
Since $\ve \gg 1$, to zeroth order, the effective Hamiltonian $\sH^\mathrm{s}$ consists of interactions of the XY type $\cos\nabla_{\be} \theta$, but only internal to each cluster type $\alpha=0,\pi$.
Hence, the zeroth order interaction describes the Hamiltonian of independent XY models defined on the distinct random phase clusters $\alpha=0$ and $\alpha=\pi$, i.e.,
\begin{equation}
    -\sum_{e\in E:\nabla_{\be} \alpha =0} \cos \nabla_{\be} \theta  =  \sH^\mathrm{XY}_{\alpha=0} (\vartheta^+) + \sH^\mathrm{XY}_{\alpha=\pi}(\vartheta^-)
\end{equation}
Where $\vartheta^\pm: \{\alpha=0,\pi\} \to (-\pi,\pi)$ are the restricted spin configurations of $\theta$ on the distinct random phase clusters $\alpha=0,\pi$.
If $p_c^\text{site} < \min (p,1-p)$, then both $\alpha=0, \pi$ subgraphs have unique infinite clusters \cite{burton1989density}. 
If $d\geq 3,$ the restrictions of $\vartheta^\pm$ to their respective infinite clusters  each undergo a phase transition at finite temperature (independent of $\ve$). 
Notably, this statement was recently demonstrated rigorously in Ref. \cite{dario2023phase}.

 Precisely, for finite boxes $\Lambda_N \subseteq \dZ^d$ centered at $0$ and having sidelength $N$, let
    \begin{equation}
        \sH^\mathrm{XY}_{\Lambda_N,\alpha} (\theta)= \sum_{e=xy\in E}  \cos (\nabla_{\be} \theta) \mathbf 1\{\alpha_x,\alpha_y =0\}
    \end{equation}
Note that, similar to our earlier  discussion, the thermodynamic limit $\bra \cdots \ket_{\dZ^d,\alpha,\beta}$ exists $\alpha$-a.s. due to compactness arguments and the Ginibre inequality.
    
    \begin{theorem}[XY ordering on site percolation clusters \cite{dario2023phase}]
    \label{thm:dario-main}
    If $d\ge 3$ and $p > p_c^\mathrm{site}(\dZ^d)$, then there exists temperature $\beta_0^\mathrm{XY}(d,p)$ such that if $\beta \ge \beta_0^\mathrm{XY}(d,p)$, then there is a constant $c>0$ so that
    \begin{equation}
      \dE_p \bra \cos(\theta_x -\theta_y)\ket_{\dZ^d,\alpha,\beta}^\mathrm{XY}\geq c>0
    \end{equation}
\end{theorem}
Thus with respect to $\sH^\mathrm{s}$ in Eq. \eqref{eq:H-strong},  we expect that for $T\ll 1$, depending on $d,p$ but independent of $\ve$, the $\alpha=0,\pi$ infinite clusters possess internal LRO.  
This suggests that we parametrize each of them (and, incidentally, each of the finite clusters as well) by one $O(2)$ spin per cluster corresponding to the direction of symmetry breaking in each one. 
It is then a question of how these effective degree of freedom orient with respect to one another.  
The latter is determined by the $O(1/\ve)$ term in $\sH^\mathrm{s}$ in Eq. \eqref{eq:H-strong}. 
The interaction $(\nabla_x^\alpha \cdot \cos \nabla \theta )^2$ only includes edges $e=xy$ with nonzero random phase gradient $\nabla_{\be} \alpha \ne 0,$ once again indicating the importance of the interfacial interaction between the $\alpha=0,\pi$ clusters.
Due to the square, this relative orientation can be either aligned or antiligned between clusters, producing the $\Z_2$ degeneracy claimed in Theorem \eqref{thm:infinite-stability}.  Finally, the stability in $(h, \beta)$ claimed in our theorem follows from the fact that, no matter how small the first order interaction strength $1/\ve$ is, it is extensive vis-\'a-vis the pair of $\alpha=0,\pi$ infinite clusters.

This discussion raises the interesting question as to the fate of our stability result for two dimensional systems.  Certainly, with respect to the standard lattice $\dZ^2,$ since $p_{c}^\mathrm{site}(\Z^2)> 1/2$ we do not expect the phenomenon to  persist.  
On the other hand, by adding, say, next nearest neighbors to the underlying graph, so that it is two dimensional but non-planar, we can arrange that $p_{c}^\mathrm{site}(\tilde{\dZ}^2)< 1/2$.  
The question then is whether the BKT phase, which still exists within each of the two infinite clusters, can disrupt the independence of $\beta_0$ from $h$.  Theorem \eqref{thm:infinite-quasi-stability} shows that, provided we weaken the definition of the ordering, stability persists.

\subsubsection{Numerics}
Since we will only be studying the effective Hamiltonian $\sH^\mathrm{s}$ rigorously in this article, it's worth performing some numerics to justify our approximation.
One check on the consistency of our heuristic is the prediction that the critical transition temperature of the effective Hamiltonian $\sH^\mathrm{s}$ in Eq. \eqref{eq:H-strong} should approximate that of the bilayer Hamiltonian $\sH$ in Eq. \eqref{eq:H-alpha} as $\ve\to \infty$.
Indeed, as shown in Fig. \ref{fig:numerics}, the transition temperature of $\sH$ on $\dZ^3$ (red dots) and that of $\sH^\mathrm{s}$ (blue triangles) seem to track each other as $\ve\to \infty$.

For completeness, we outline that the relatively standard procedure to determine the transition temperature of a discrete $\dZ_2$ symmetry using the Binder cumulant (Eq. (3.17) of Ref. \cite{binder2022monte}).
For the bilayer Hamiltonian, one performs the Metroplis-Hasting algorithm\footnote{We note that since the Hamiltonians $\sH,\sH^\mathrm{s}$ are local in space, a checkerboard update method can be applied when utilizing Metropolis-Hasting, instead of single site spin flips. Since the systems are disordered, it's also practical to utilize parallel tempering \cite{marinari2007optimized} to speed up convergence.} (see, e.g., Chapter 18 and 20 of Ref. \cite{klenke2013probability}) on finite size torus $\dZ_N^3$ (with $N=8,12,16,20$) for a given random phase $\alpha$ at disorder strength $\ve$ and inverse temperature $\beta$, from which we approximate
\begin{equation}
     \bra [\sin\phi]_{\dZ_N^3}^2\ket_{\dZ_N^3,\alpha,\ve,\beta}, \bra [\sin\phi]_{\dZ_N^3}^4\ket_{\dZ_N^3,\alpha,\ve,\beta},  \quad [\sin \phi]_{\dZ_L^3} \equiv \frac{1}{|\dZ_N^3|} \sum_{x\in \dZ_N^3} \sin \phi_x
\end{equation}
via Birkhoff's theorem \cite{birkhoff1931proof}.
Note that $[\sin \phi]_{\dZ_N^3}$ is the $\dZ_2$ order parameter averaged over the finite torus.
We then average over (in our case, 40) different samples of random phases $\alpha$ to approximate
\begin{equation}
    m_{N,\ve,\beta}^{(2)} \equiv \dE_p  \bra [\sin\phi]_{\dZ_N^3}^2\ket_{\dZ_N^3,\alpha,\ve,\beta},  \quad m_{N,\ve,\beta}^{(4)} \equiv \dE_p \bra [\sin\phi]_{\dZ_N^3}^4\ket_{\dZ_N^3,\alpha,\ve,\beta}
\end{equation}
And then compute the Binder cumulant
\begin{equation}
    U_{N,\ve,\beta} = 1-\frac{m_{N,\ve,\beta}^{(4)} }{3 m_{N,\ve,\beta}^{(2)} }
\end{equation}
It's expected that at high $T\gg 1$, the Hamiltonian approximates that of a Gaussian and thus the order parameters satisfies Wick's theorem so that $U_{N,\ve,\beta}\sim 0$, while at low $T\ll 1$, the order parameter exhibits LRO and thus $U_{N,\ve,\beta} \sim 2/3$.
For distinct pairs of $N,N'$, we compute the ratio $U_{N,\ve,\beta}/U_{N',\ve,\beta}$ and find the temperature $\beta_{(N,N')}(\ve)$ at which the ratio is closest to $1$. 
The critical transition $\beta_c(\ve)$ is then determined by averaging over distinct pairs $(L,L')$ and the error bars in Fig. \ref{fig:numerics} is due to the sample variance.
A similar procedure is done for the effective Hamiltonian $\sH^\mathrm{s}$ but instead of $\sin \phi$, we use $\tau$ as defined in Theorem \eqref{thm:infinite-stability}, where the two are connected by the heuristics in Eq. \eqref{eq:heuristic}.

\begin{figure}[ht]
\centering
\includegraphics[width=0.5\columnwidth]{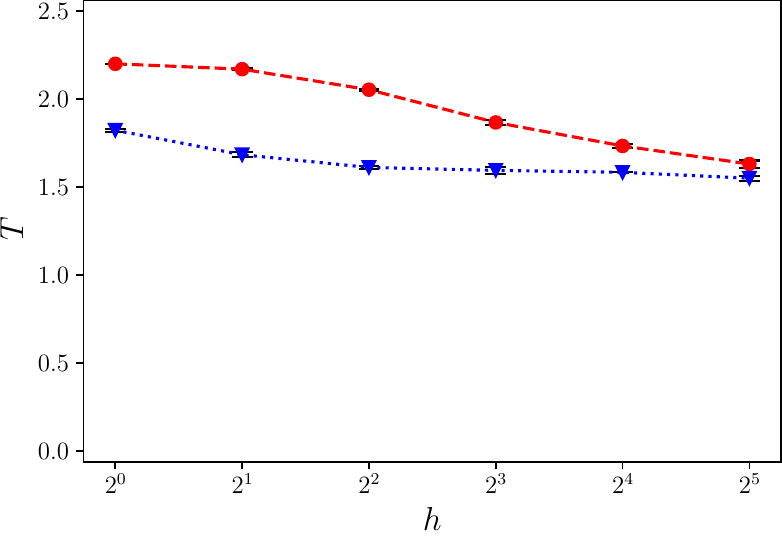}
\caption{Numerics in $\dZ^3$ with $p=1/2$. The red/blue depict the critical transition temperature of the bilayer/effective Hamiltonian $\sH,\sH^\mathrm{s}$ defined in Eq. \eqref{eq:H-alpha} and \eqref{eq:H-strong}, respectively. Error bars are also provided based on the Binder cumulant.
}
\label{fig:numerics}
\end{figure}

\subsection{Heuristics for weak disorder.}
\label{sec:weak-disorder}
Although we are primarily concerned with the strong disorder regime $\ve \gg 1$, it is instructive to understand and compare current results of the weak disorder regime $\ve \ll 1$.
When $\ve \ll 1$, it has been heuristically argued \cite{yuan2023inhomogeneity} that the average phase $\theta \equiv (\theta^++ \theta^-)/2$ and the phase difference $\phi \equiv \theta^+ -\theta^-$ decouple so that the phase transition with respect to $\phi$ can be qualitatively described by the following \textit{effective} Hamiltonian
\begin{equation}
    \label{eq:H-weak}
    \sH^\mathrm{w}_{G,\alpha,\ve} =2\sH^\mathrm{XY}_G (\phi) -\ve \sum_{x\in V} \cos (\phi_x -\alpha_x).
\end{equation}
The heuristic argument is relatively crude: in the weak disorder regime $h\ll 1$ and for low temperatures $T\ll 1$, 
we assume $|\nabla_{\be} \theta^\pm| \ll 1$ for all edges $e$ (or $|\nabla_{\be} \sigma^\pm| \ll 1$ in original XY spins) and then write 
\begin{align}
\label{eq:weaktrans}
    \sH &\sim  \frac{1}{2} \sum_{e} (\nabla_{\be}\theta^+)^2+(\nabla_{\be} \theta^-)^2  -\ve \sum_{x\in V} \cos (\phi_x-\alpha_x) \\
    &\sim \frac 12 \underbrace{\sum_{e}[1- \cos (\nabla_{\be} \theta)]}_{\sim \sH^\mathrm{XY}(\theta)} +\frac 12 \underbrace{\sum_{e} [1- \cos(\nabla_{\be} \phi)] -\ve \sum_{x\in V} \cos (\phi_x-\alpha_x)}_{\sim \sH^\mathrm{w}(\phi)}
\end{align}
The two terms on the RHS decouple. One can then ignore the first term and focus on the second. One might worry here that the change of variables $\theta, \phi$ is not globally well defined, but this misses the point.  The effective Hamiltonian makes sense and provides a good proxy for the behavior, in a fixed region $R$, for the restricted collection of spin configurations for which the Dirichlet energy behaves as
\beq
E_R(\sigma):=\sum_{e} \frac{1}{2} \left((\nabla_{\be}\theta^+)^2+(\nabla_{\be} \theta^-)^2 \right)=O(h|R|).
\eeq
At low temperatures, this is typical for most regions.  Furthermore, provided $h R$ is not too big, the bound on $E_R$ will imply $(\theta, \phi)$ are well defined within $R$. 

In the case where $p=1/2$, the second term $\sH^\mathrm{w}$ is exactly the RFO(2) model discussed in Ref. \cite{crawford2013random,crawford2014random,crawford2024random}.  In those papers, it was proved that provided that $\ve$ is sufficiently small, there is SSB with a $\dZ_2$ degenerate GS space characterized by $e^{i\phi} \sim^\mathrm{w}\pm i$.
To formulate the result precisely, let 
$\bra \cdots \ket^{\mathrm{w}}_{\partial,\Lambda_N,\alpha,\ve,\beta}$ denote the finite volume Gibbs measure induced by $\sH^\mathrm{w}$ on the finite box $\Lambda_N$ with fixed boundary condition $e^{i\phi_x} = +i$ for all $x\in \partial \Lambda_N \equiv \Lambda_{N}  \backslash \Lambda_{N-1}$.
\begin{theorem}[Weak disorder \cite{crawford2013random,crawford2014random,crawford2024random}]
\label{thm:weak}
    Let $d=2,3$ and $p=1/2$ and $\xi\in (0,1)$ be fixed and sufficiently small. There exists $\ve_0(\xi) >0$ such that given $\ve\in (0,\ve_0)$, there exists $\beta_0 (d,p,\ve)$ depending on $\ve$, such that if $\beta  \ge \beta_0(d,p,\ve)$, then
    \begin{equation}
        \dE_{p=1/2}\left[\limsup_{N\to \infty} \left| \bra \sigma_{\Lambda_N} \ket^{\mathrm{w}}_{\partial,\Lambda_N,\alpha,\ve,\beta} -i \right|\right] \le \xi, \quad \sigma_{\Lambda_N} \equiv \frac{1}{|\Lambda_N|} \sum_{x\in \Lambda_N} e^{i\phi_x},
    \end{equation}
\end{theorem}
\noindent
The constraint on dimensionality here is technical.  
Applied to $\sH,$ this result suggests that the pair $(\sigma^+_x, \sigma^-_x)$ should be locked at right angles to one another in space, with fast, low amplitude oscillations in the angle due to the phase shift $\alpha_x:$ if $\alpha_x=0,$ $(\sigma^+_x, \sigma^-_x)$ tilt toward each other by angle of order $O(h)$ whereas if $\alpha_x=\pi,$ they tilt slightly away from each other.

Adapting the proof of Theorem \eqref{thm:weak} to handle $\sH$ when $d=3$ seems plausible, if challenging. The issue is whether or not proving LRO for $\sigma^++\sigma^-$ is a technical necessity to control the phase locking of $\phi.$  This hurdle is even higher in $d=2$ due to the expected BKT phase satisfied by $\sigma^++\sigma^-$ at low temperature.  Nevertheless, we believe that this phase locking does occur for all $d\geq 2$ in the weak disorder regime.

In light of our previous emphasis on the independence of the transition temperature from the field strength for strong disorder, it is worth reflecting on its behavior at weak disorder.  A careful review of the proof of Theorem \eqref{thm:weak} provides a lower bound on the critical temperature $T_\mathrm{LRO}^\mathrm{w}(\ve) \ge T_0(\ve) \sim \ve^2$  (up to logarithm corrections with respect to $\ve$), though this was not optimized there.
On the other hand, mean-field theory computations, see \cite{yuan2023exactly}, 
predict a disorder strength $\ve$ independent transition temperature on the same scale as the XY coupling constant, i.e., $T_\mathrm{LRO}^\mathrm{w}(\ve)  \sim 1$ for $\ve \ll 1$.
Extrapolating to $d = 3,$ at zero disorder $\ve =0$ the model reduces to the clean XY model and thus exhibits LRO/SSB below a critical transition temperature $T_\mathrm{LRO}^\mathrm{w}(\ve=0) \sim 1$ \cite{frohlich1976infrared,garban2022continuous,friedli2017statistical}.
Hence, for weak disorder $\ve \ll 1$, it is then natural to conjecture that $T_\mathrm{LRO}^\mathrm{w}(\ve) \sim 1$.

This question is more intriguing in $d=2$ dimensions, due to the BKT phase for the clean XY model. 
If a phase transition in $\phi$ occurs at an $O(1)$ temperature for all $h$ small, there might be multiple critical points, correpsonding to a first transition from disordered to quasi-LRO and then to a true LRO phase at a much lower temperature $T_\mathrm{LRO}^\mathrm{w}(\ve) \sim \ve^2$.
While we will not discuss this question further in the weak disorder regime, a similar question arises for strong disorder, 
see Sections \eqref{sec:open-questions} and \eqref{sec:H-clean}.

\subsection{Open Questions}
\label{sec:open-questions}

Here we list some important open questions. Some of them will be partially answered later in this article, while others we will leave for future investigation.
\begin{OQ}
Can analagous 2D systems be infinitely stable under disorder, instead of merely infinitely quasi-stable?
\end{OQ}

In $d=2$ with nearest-neighbor couplings, the previous argument fails since $p_c^\text{site}(\dZ^2) >1/2$ implies that the individual clusters are finite and exponentially decaying in their size distribution \cite{Grimmettbook}.
Adding next nearest-neighbor connections so that $p_c^\text{site}(\tilde{\dZ}^2) < 1/2$, however, only guarantees quasi-LRO $T_\text{QLRO} \sim 1$.
Indeed, each infinite cluster on $\tilde{\dZ}^2$ exhibits a BKT phase and thus whether LRO exists at $\ve$ independent temperatures remains an open question.

Even for a system \textit{without} disorder (we shall refer to such a setup as clean), this question is not trivial.
Consider the following further simplified model.
Let $G=\dZ^2,$ and set 
\begin{equation}
    \label{eq:H-clean}
    \sH^\mathrm{c}(\vartheta^\pm)= \sum_{\ell =\pm} \sH^\mathrm{XY}(\vartheta^\ell) -\frac{1}{2\ve} \sum_{x} \cos^2 \varphi_x,
\end{equation}
where each vertex gets a pair of XY spins, indexed by $\vartheta^\pm_x$. 
The first term on the RHS denotes that each layer interacts via a ferromagnetic XY interaction. 
In the second term, $\cos^2 \varphi_x$ denotes the interlayer coupling between the spins at $x$ where $\varphi_x = \vartheta^+_x -\vartheta_x^-$.
The square is inserted to give the model a $\dZ_2$ degeneracy;
it is a simple proxy for the first order interaction $(\nabla_x^\alpha \cdot \cos \nabla \theta )^2$ in the effective Hamiltonian $\sH^\mathrm{s}$ in Eq. \eqref{eq:H-strong}.  
Clearly the GS space for $\sH^\mathrm{c}$ has $\dS^1\times \dZ_2$ degeneracy, i.e., the layers are either fully aligned or fully anti-aligned for all coupling strengths $\ve$. 
Since the terms in $\sH^\mathrm{c}$ satisfy the necessary conditions for Ginibre's inequalities \cite{ginibre1970general,tokushige2024graphical}, the critical transition temperature at finite $1/\ve$ is bounded below by that at $1/\ve =0$, i.e., the clean XY models defined on $\dZ^2$.
Since the latter exhibits quasi-LRO below a temperature independent of $\ve$, the same holds true for the clean Hamiltonian $\sH^\mathrm{c}$ on $\dZ^2$ at arbitrarily weak interlayer coupling $1/\ve$, i.e., $T_\mathrm{QLRO}^\mathrm{c} \sim 1$, analogous to the effective Hamiltonian $\sH^\mathrm{s}$ on $\tilde{\dZ}^2$ (with $p_c^\mathrm{site} (\tilde{\dZ}^2) < \min(p,1-p)$).

Due to the similarities between the models, it then raises the natural question whether the clean system $\sH^\mathrm{c}$ on $\dZ^2$ can exhibit LRO at arbitrarily weak coupling $1/\ve$ below a coupling independent temperature $T_\mathrm{LRO}^\mathrm{c} \sim 1$.
From this perspective, field theory \cite{Seibergetal2021} suggests that the $\cos^2 \varphi$ interaction in Eq. \eqref{eq:H-clean} is always relevant, in the renormalization group sense,  and thus gives rise to LRO $T_\mathrm{LRO}^\mathrm{c} \sim 1$ in the limit $1/\ve \to 0$.
Numerics \cite{song2022phase,bojesen2014phase,maccari2022effects} seem to corroborate this conjecture and show only a single phase transition for $\varphi$ - a disordered phase at high $T$ and LRO below $T_\mathrm{LRO}^\mathrm{c} \sim 1$. 
Nevertheless, the mathematical evidence regarding the phase diagram is not clear cut.
For related $\dZ_2$ invariant systems in 2D, Ref. \cite{shlosman1980phase} shows that an analog to $\varphi$ exhibits LRO for temperatures below $\sim 1/\ve$ so that $T_\mathrm{LRO}^\mathrm{c}\gtrsim 1/\ve$ is bounded below (The proof utilizes reflection positivity and thus can be straightforwardly adapted to $\sH^\text{c}$ in Eq. \eqref{eq:H-clean}).
In this paper, we further prove $\varphi$ can only order along $0,\pi$ (no ordering of $\sin \varphi$) with the following upper bound, i.e., if $T\ll 1$ (independent of $\ve$), then
\begin{equation}
    \label{eq:upper-bound}
    |\bra \cos \varphi_x \ket| \lesssim \exp\left(-\frac{1}{8\pi} \frac{T}{|\log \ve|^{-1}}\right).
\end{equation}
These results are collected in Sec. \eqref{sec:H-clean} (or more specifically, Corollary. \eqref{cor:absence-Y} and Theorem. \eqref{thm:mcbryan}.
Notably, if  $\ve$ is sufficiently large so that $1/|\log \ve| \ll 1 \sim T_\text{QLRO}$, then Eq. \eqref{eq:upper-bound} implies that if LRO occurs at $1\sim T_\text{LRO}\le T_\text{QLRO}$, the spontaneous magnetization must be exponentially small with respect to $|\log \ve|$.
In comparison (see Fig. \ref{fig:upper-bound}), universality ``conventionally" suggests that the magnetization $m$ 
becomes substantial at temperatures of the same scale as the transition temperature, i.e., $m \sim 1$ for $T\sim T_\text{LRO}$.

Based on the previous discrepancy, it appears that LRO occurs at best at $T_\text{LRO} \sim 1/|\log \ve|$, suggesting the possibility of an intermediate massless phase (for both the clean model in Eq. \eqref{eq:H-clean} and the disordered model in Eq. \eqref{eq:H-alpha}) in the limit $1/\ve \to 0$.
One may then speculate that such an intermediate phase has not been observed in numerics \cite{song2022phase,bojesen2014phase,maccari2022effects} due to the extremely slow logarithmic decay.

\begin{figure}[ht]
\centering
\includegraphics[width=0.5\columnwidth]{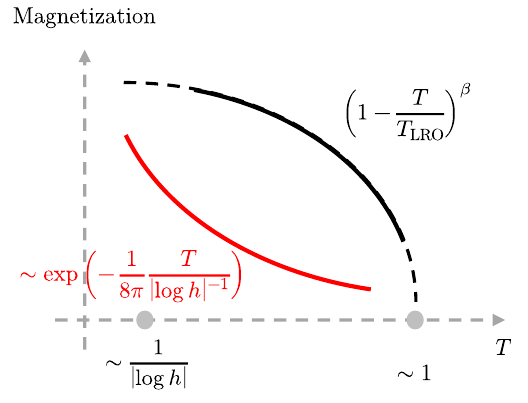}
\caption{Upper bound. The $x$ axis denotes the temperature, while the $y$ axis denotes an order parameter such as magnetization. The red line denotes the upper bound proved in Eq. \eqref{eq:upper-bound}. If $T_\text{LRO}\sim 1$, then universality implies that the magnetization follows the form of the black line with critical exponent $\beta$.
}
\label{fig:upper-bound}
\end{figure}

\begin{OQ}
Can the analogous system be infinitely stable in the presence of $U(1)$ disorder, i.e., $\alpha_x \in (-\pi,\pi]$? Or more importantly, can the system with $U(1)$ disorder exhibit a finite temperature phase transition in low dimensions $d\ge 4$ (and thus provide a counter-example of sorts to the Imry-Ma phenomenon)?
\end{OQ}

As discussed previously, the Imry-Ma phenomenon \cite{imry1975random} claims that in the presence of random fields, the phase transition of spin systems is destroyed in low dimensions ($d\le 2,4$ for discrete and continuous symmetries).
Although the statement has been proven rigorously for selected examples such as the random field Ising and XY models \cite{aizenman1989rounding,aizenman1990rounding} and their quantum analogues \cite{greenblatt2009rounding,aizenman2012proof}, it is generally believed that the phenomenon persists beyond these selected examples.
One major reason is due to the belief that finite temperature ordering can only occur in the weak disorder regime.
In the weak disorder regime, any reasonable lattice model can be heuristically simplified using the nonlinear sigma representation (as discussed in Sec. \eqref{sec:weak-disorder}), upon which the Imry-Ma argument can be  applied.
However, since we have now shown that it is possible for finite temperature ordering to persist in the infinite disorder limit, it raises the question whether one can find a counter-example to the Imry-Ma phenomenon.

Unfortunately, our bilayer Hamiltonian in Eq. \eqref{eq:H-alpha} does not serve as a counter-example, since, similar to the RFO(2) model, ordering of the $\dZ_2$ degenerate GS occurs perpendicular to the submanifold on which the random field disorder acts.
More concretely, the random phase $\alpha_x = 0,\pi$ corresponds to a coupling along the $X$ direction, while the $\dZ_2$ degenerate GS $e^{i\phi} \sim \pm i$ are ordered in the $Y$ direction, and thus one may regard the ordering along the $X$ direction destroyed by the disorder.
However, this begs the interesting question of whether a slight variation (bilayer Hamiltonian with $U(1)$ disorder, i.e., uniformly distributed $\alpha \in (-\pi,\pi]$) can give rise to distinct behavior, especially in the strong disorder regime $\ve\gg 1$.

If one attempts to repeat the previous heuristics and map the bilayer Hamiltonian with $U(1)$ disorder onto an effective Hamiltonian in the strong disorder regime $\ve\gg 1$, it is straightforward to check that depending on the random phase $\alpha$, certain interactions of the effective Hamiltonian will be anti-ferromagnetic and thus Ginibre inequality cannot be applied.

%% file: proof.tex
\section{Definition and Preliminaries}
\subsection{Preliminaries}

When considering infinite, we shall always mean countably infinite.
\begin{definition}
    Given an infinite graph $G_\infty$ with vertices $V_\infty$ and edges $E_\infty$, we say $x,y\in V_\infty$ are \textbf{adjacent} and write $x\sim y$ if $e=xy$ is an edge.
    Similarly, edge $e$ and site $x$ are \textbf{adjacent} and write $e\sim x$ if $e=xy$ for some site $y$.
    Similarly, edge $e$ and subset $V\subseteq V_\infty$ are \textbf{adjacent} and write $e\sim V$ if $e\sim x$ for some $x\in V$.
    Moreover, given a subset $V\subseteq V_\infty$, the subgraph induced by $V$ is that consisting of edges $e=xy\in E_\infty$ such that $x,y\in V$, and when there is no confusion, we will also use $V$ to denote the induced subgraph.
\end{definition}
\begin{definition}
Slightly abusing terminology, we say that a finite box $\Lambda\subseteq \dZ^d$ has \textbf{length} $L$ if $\Lambda=i+\Lambda_L$ for some $i\in \dZ^d$ and $\Lambda_L =\{-L,...,L\}^d$. We write $\Lambda =[i]_L$ for simplicity.
    The \textbf{boundary} of $\Lambda_L$ is $\partial \Lambda_L =\Lambda_L \backslash \Lambda_{L-1}$
\end{definition}


\subsection{Disordered Models on $\dZ^d$ and $\tilde{\dZ}^2$}

\begin{definition}
    We denote the graph $\dZ^d,d\ge 2$ consisting of sites on $\dZ^d$ and the conventional nearest-neighbor edges, i.e., $e=xy$ is an edge if $\norm{x-y}_1=1$. 
    We denote the graph $\tilde{\dZ}^2$ consisting of sites on $\dZ^2$ and edges involving nearest-neighbors and next nearest-neighbors, i.e., $e=xy$ is an edge if $\norm{x-y}_\infty = 1$.
    We note that the site percolation threshold satisfies 
    \begin{align}
    p_c^\mathrm{site}(\dZ^d) &< 1/2, \quad d \ge 3,\\
    p_c^\mathrm{site}(\tilde{\dZ}^2) &< 1/2 < p_c^\mathrm{site}(\dZ^2).
    \end{align}
\end{definition}

\begin{definition}
    Let $G_\infty=(V_\infty,E_\infty)$ denote an infinite graph with vertices $V_\infty$ and edges $E_\infty$.
    A \textbf{random phase}  $\alpha:V_\infty \to \{0,\pi\}$ is a field of i.i.d.~Bernoulli random variables on sites $V_\infty$ with probability $p$ such that $\alpha_x=0$. 
    We refer to sites $x\in V_{\infty}$ as \textbf{occupied} and \textbf{vacant} if $\alpha_x=0,\pi$ respectively. The reader should keep in mind that we will be considering the percolative properties of BOTH occupied and vacant sites. 


    {
    For vertices $a,b\in S\subseteq V_\infty$, we write $a\xleftrightarrow[S]{\pm} b$ if $a,b$ are connected in $S$ using, respectively, the occupied/vacant sites of random phase $\alpha$ in subset $S$.
    Similarly, we write $A\xleftrightarrow[S]{\pm} B$ for subsets $A,B\subseteq S$ if there exists $a\in A,b\in B$ such that $a\xleftrightarrow[S]{\pm} b$.
    In particular, we slightly abuse notation and write $e \xleftrightarrow[S]{\pm} B$ for edge $e\in E_\infty$ by treating $e$ as a set of its two endpoints.
    In the case where $S=V_\infty$, we shall omit the subscript $S$ and write $A\xleftrightarrow{\pm} B$ instead.
    }
\end{definition}

\begin{definition}[Disordered Model]
    \label{def:H-strong}
    Given a random phase $\alpha$ on infinite graph $G_\infty$, define the \textbf{effective model (in the strong disorder limit)} with inverse temperature $\beta$ via the finite-volume Gibbs measure on finite subgraphs $G$ induced by the finite subset $V\subseteq V_\infty$ over \textbf{spin configurations} $\theta:V\to (-\pi,\pi]$,
    \begin{equation}
        \mu_{G,\ve,\beta}^{\mathrm{s}, \alpha}[d\theta] = \frac{1}{\sZ_{G,\ve,\beta}^{\mathrm{s}, \alpha}} \exp\left(-\beta \sH_{G,\alpha,\ve}^\mathrm{s}(\theta)\right) \prod_{x\in V} d\theta_x,
    \end{equation}
    where $\sH_{G,\alpha,\ve}^\mathrm{s}$ is given in Eq. \eqref{eq:H-strong}, $d\theta_x$ denote uniform measures on $(-\pi,\pi]$ and $\sZ^{\mathrm{s}, \alpha}_{G,\beta,\ve}$ is the \textbf{partition function} which normalizes the measure.
    We write $\bra \cdots \ket^{\mathrm{s},\alpha}_{G,\ve,\beta}$ to denote the expectation.
\end{definition}

\subsection{Clean Models}
\begin{definition}(Extended graphs $\dZ^d|_n$)
    Following Ref. \cite{dario2023phase}, we shall also consider special subgraphs of $\dZ^d$, referred to as \textbf{extended subgraphs} $\dZ^d|_n$, whose sites consists of the vertices $x=(x_1,...,x_d)\in \dZ^d$ such that at least $d-1$ of the components $x_i$ satisfy $x_i \in n\dZ$.  The remaining component is unconstrained. The graph structure  on $\dZ^d|_n$ is induced by its natural embedding in $\Z^d$, i.e. $e=xy$ is an edge in $\dZ^d|_n$ if $x,y$ are sites in $\dZ^d|_n$ and $xy$ is an edge in $\dZ^d$.
    We  write $n|x$  if all components $x\in \dZ^d$ are divisible by $n$ and let $n|S$ denote the collection of all such sites $x\in S$ for any subset $S\subseteq \dZ^d$.
\end{definition}

 \begin{definition}[Clean Model]
     \label{def:H-clean}
     Let $G_{\infty}=(V_{\infty}, E_{\infty})$ be an infinite graph with uniformly bounded degree.  Let $E_1, E_2\subseteq E_{\infty}$.   Let $(G, E)$ be a finite subgraph.  We define the associated \textbf{clean model} with inverse temperature $\beta$ and \textbf{interlayer interaction} $1/\ve$ as  the finite-volume Gibbs measure  with Hamiltonian 
     \begin{equation}
         \sH_{G, h}^\mathrm{c} (\vartheta) = -\sum_{e\in E_1\cap E} \cos(\theta_x-\theta_y) -\frac{1}{\ve} \sum_{e\in E_2\cap E}  \cos(\theta_x-\theta_y) ^2
     \end{equation}
     \end{definition}
     It is worth noting here that as long as $h, \beta>0,$  by Ginibre's inequality, thermodynamic limits always exist (free boundary conditions).  We write $\bra \cdots \ket_{G_{\infty},\ve,\beta}^\mathrm{clean}$ for the infinite volume expectation.  

\begin{definition}
\label{def:H-clean}
Let $G_\infty^\pm=(V_\infty^\pm,E_\infty^\pm)$ denote a pair of infinite graphs and let $\sI_\infty$ denote an infinite collection of unordered pairs $x^+x^-$ where $x^\pm \in V_\infty^\pm$, which we refer to as the \textbf{interlayer edges}.
Let $V^\pm \subseteq V_\infty^\pm$ denote finite subsets with induced subgraphs $G^\pm$ and $\sI$ denote the subset of interlayer edges $x^+x^-$ with $x^\pm \in V^\pm$. 
Let $\sG$ be the (bipartite) graph obtained from $G^+\cup G^-$ by including the edges  $\sI$ between vertices of $G^+, G^-$ and similarly define $\sG_\infty$. 
The \textbf{(bilayer) clean model} with inverse temperature $\beta$ and \textbf{interlayer interaction} $1/\ve$ is characterized by the finite-volume Gibbs measure associated to the Hamiltonian
\begin{equation}
    \sH_{\sG, h}^\mathrm{c} (\vartheta^\pm) = 2\sum_{\ell =\pm} \sH^\mathrm{XY}_{G^\ell} (\vartheta^\ell) -\frac{1}{\ve} \sum_{ e\in \sI} \cos^2 (\varphi_{\be})
\end{equation}
where $\sH^\mathrm{XY}_{G^\pm}$ is the XY model on $G^\pm$ given in Eq. \eqref{eq:H-XY} and $\varphi_{\be} =\vartheta^+_{x^+} -\vartheta^-_{x^-}$ is the \textbf{interlayer} phase difference on $e=x^+x^-\in \sI$ with orientation $\be$ from $x^-$ to $x^+$.
We then let $\bra \cdots \ket^\mathrm{c}_{\sG_\infty,h,\beta}$ denote the expectation (whose existence is guaranteed by Ginibre).
In particular, we say that the clean model is defined on
\begin{itemize}
    \item $\dZ^d|_n$ if $G^\pm_\infty =\dZ^d|_n$ and $\sI_\infty$ is the collection of unordered pairs $x^+x^-$ where $x^+, x^-$ are the two copies of $x\in n|\dZ^d.$ 
    With slight abuse of notation, we write $\sG_\infty = \dZ^d|_n$
    \item $\dZ^d|_n$ with \textbf{random environment} $r:n|\dZ^d \to \{0,1\}$ if $G^\pm_\infty = \dZ^d|_n\backslash \{r=0\}$ and $\sI_\infty$ is the collection of unordered pairs $x^+x^-$ where $x^+, x^-$ are the two copies of $x\in (n|\dZ^d) \backslash \{r =0\}$ and $\{r=0\}$ is the set of all $x\in n|\dZ^d$ with $r_x=0$. With slight abuse of notation, write $\sG_\infty = \dZ^d|_n^r$
\end{itemize}
{In either case, we simplify notation and write $x\in \sI_\infty$ and $\varphi_x$ since the interlayer edges $e=x^+x^-\in \sI_\infty$ involve two copies $x^+,x^-$ of $x$.}
\end{definition}

We also need the following variation on the previous definition.

\begin{definition}[Clean model with multiple coupling strengths]
    \label{def:H-clean-multiT}
    Consider the clean model on $\sG_\infty = \dZ^d|_n$ or $\dZ^d|_n^r$ with random environment $r$ from the previous definition. 
    Let $\Lambda_L \subseteq \dZ^d$ be a finite box with $L=Nn$, $\Lambda_L^\pm$ the induced subgraph of $G_\infty^\pm$ and  $\sI_L$ denote the induced interlayer edges. 
    Write $\sG_L$ as the bipartite graph obtained from $\Lambda_L^+\cup \Lambda_L^-$ by including the edges $\sI_L$.
    The \textbf{clean model} on $\sG_L$ with couplings $J,J'$ is defined by the finite-volume Gibbs measure on finite boxes $\Lambda_L \subseteq \dZ^d$
    \begin{align}
        \mu_{\sG_L,\ve, J,J'}^\mathrm{c}[d\vartheta^\pm] &= \frac{1}{\sZ_{\sG_L,\ve, J,J'}^\mathrm{c}} e^{-\sH_{\sG_L,h,J,J'}(\vartheta^\pm)}\prod_{\ell=\pm, x\in \Lambda_{L}} d\vartheta_x^\ell, \text{ where}\\
        -\sH_{\sG_L,h, J , J'}(\vartheta^\pm) &=2\sum_{\ell=\pm}\left[J\sum_{e\in \Lambda_L^\ell:e\sim \sI_L} \cos \nabla_{\be} \vartheta^\ell+J'\sum_{e\in \Lambda_L^\ell:e\not\sim \sI_L}\cos \nabla_{\be} \vartheta^\ell\right] + \frac{J}{\ve} \sum_{x\in \sI_L} \cos^2 \varphi_{x}
    \end{align}
    $\sZ_{\sG_L,\ve,J,J'}^\mathrm{c}$ is the partition function which normalizes the measure, $d\vartheta_x^\ell$ denotes the uniform measure and $\varphi_x \equiv \vartheta_x^+ -\vartheta_x^-$ denotes the phase difference along the interlayer edges $\sI_L$.
    We emphasize that the random environment $r$ enters implicitly via $\sG_\infty$. 
\end{definition}

\begin{prop}[Proposition (2.17) of Ref. \cite{dario2023phase}]
    \label{prop:H-clean}
    Consider the clean model from  Definition \eqref{def:H-clean-multiT} with $r\equiv 1$.
    Then there exists finite $J_0^\mathrm{c}(d)<\infty$ and $J_0^{\mathrm{c}\prime}(d,n)<\infty$ such that if $J \ge J_0^\mathrm{c}(d)$ and $J' \ge J_0^{\mathrm{c}\prime}(d,n)$, then the correlation function
    \begin{equation}
        \bra \cos(\varphi_x)\cos(\varphi_y)\ket_{\dZ^d|_n,\ve,J,J'}^\mathrm{c}
    \end{equation}
    \begin{itemize}
        \item is bounded away from 0 as $\norm{x-y}\to \infty$ for $d\ge 3$
        \item is bounded below by an algebraically decaying function in $\norm{x-y}$ as $\norm{x-y}\to \infty$ for $d=2$.
    \end{itemize}
\end{prop}
\begin{proof}
    By Ginibre's inequality, we see that the correlation function at interaction strength $1/\ve$ is bounded below by that with interaction strength $1/h=0$.
    The latter corresponds to two independent XY models defined on $\dZ^d|_n$ and thus the statement follows from Proposition (2.17) of Ref. \cite{dario2023phase}.
\end{proof}


\subsection{General XY models and the Ginibre inequality}

The main tool we will need to relate the clean and disordered models defined above is the Ginibre inequality, which was already invoked to take thermodynamics limits.  The class of models which fall within the scope of this inequality is defined as follows.

\begin{definition}[General XY]
    \label{def:general-XY}
    Let $V_\infty$ be a collection of infinite sites.
    Let $\sM_\infty$ be the collection of integer-valued functions $m:V_\infty \to \dZ$ which have finite support, denoted $\supp m$.
    A choice of (nonegative) \textbf{coupling constants} for the general $XY$ model is any function $J:\sM_\infty \to [0,\infty)$ such that for all finite subsets $V\subseteq V_\infty$, $J_m$ is nonzero for at most finitely many  $m\in \sM_\infty$ supported in $V$.
    The \textbf{general XY model} on finite $V$ is defined by the Gibbs measure on spin configurations $\theta:V\to (-\pi,\pi]$ via 
    \begin{equation}
        \mu_{V,J}[d\theta] = \frac{1}{\sZ_{V,J}} e^{-\sH_{V,J}(\theta)} \prod_{x\in V} d\theta_x
    \end{equation}
    where $\sZ_{V,J}$ is the partition function normalizing the probability measure, $d\theta_x$ are the uniform measure on $(-\pi,\pi]$ and $\sH_{V,J}$ is the (finite-volume) Hamiltonian given by 
    \begin{equation}
        \sH_{V,J}(\theta) = -\sum_{m\in \sM_\infty\atop{\supp{m} \subseteq V}} J_m \cos (m\theta), \quad m\theta \equiv \sum_{x\in V} m_x \theta_x
    \end{equation}
\end{definition}
Both the disordered model in Definition \eqref{def:H-strong} and the clean model in Definition \eqref{def:H-clean}, \eqref{def:H-clean-multiT} are general XY models. 

For any pair $J^1, J^2$ of nonnegative couplings for the general XY model,  let us say $J^1 \leq J^2$ if $J^1_m\leq J^2_m$ for all $m\in \mathcal M_{\infty}.$

\begin{theorem}[The Ginibre inequality]
\label{T:Gin}
For any $m \in \mathcal M_{\infty}$ and any $V\subseteq V_{\infty}$ finite,
\begin{equation}
        \bra \cos(m\theta) \ket_{V,J^1} \le \bra \cos(m\theta)\ket_{V,J^2}.
\end{equation}
\end{theorem}


\subsection{Site Percolation}

We want to use the Ginibre inequality to compare  the disordered model $\sH^\mathrm{s}$ with the clean model $\sH^\mathrm{c}$ corresponding to $\Z^d|_{n}$ for $n$ large enough and  at a higher temperature.
Our procedure for constructing this embedding closely follows Ref. \cite{dario2023phase}. We collect here the percolation theoretic definitions required.
\begin{definition}
    A finite box $\Lambda \subseteq \dZ^d$ of length $L$ is \textbf{pre-good} with respect to occupied (vacant) sites of $\alpha$ in $\Lambda$ if there exists a unique cluster $C \subseteq \Lambda_L$ of occupied (vacant) sites which touches the $2d$ faces of box $\Lambda$ and the diameter of all other clusters of occupied (vacant) sites within $\Lambda$ are $\le L/100$.
\end{definition}
\begin{definition}
    A finite box $\Lambda \subseteq \dZ^d$ of length $L$ is \textbf{good} with respect to occupied (vacant) sites of $\alpha$ if $\Lambda$ is pre-good with respect to occupied (vacant) sites of $\alpha$ and all finite boxes $\Lambda'$ of linear length $L/10$ that intersect $\Lambda$ are pre-good with respect to occupied (vacant) sites of $\alpha$.
    {For simplicity, we say a finite box is \textbf{good} if it is good with respect to both occupied and vacant sites of $\alpha$.}
\end{definition}
\begin{remark}
    \label{rem:good-boxes}
    From the previous definition, it's straightforward to check that
    \begin{itemize}
        \item If $\Lambda,\Lambda'$ are adjacent good boxes with respect to the occupied (vacant) sites with the same length, then the unique clusters $C,C'$ are connected.
        \item If $\Lambda$ is a box of length $L$, then the number of boxes $\Lambda'$ with length $L/10$ that intersect $\Lambda$ and are pre-good is polynomial in $L$.
    \end{itemize}
\end{remark}
A classical result regarding these definitions is the following.
\begin{theorem}[\cite{pisztora1996surface,penrose1996large}]
    \label{thm:good-box}
    Let $p>p_c^\mathrm{site}(G_\infty)$ for {$G_\infty =\dZ^d$ or $\tilde{\dZ}^2$}. Then there exists constant $C,c>0$ depending on $p,d$ such that for any box $\Lambda\subseteq \dZ^d$ of length $L$,
    \begin{align}
        \dP_p [ \Lambda\text{ is not pre-good with respect to occupied sites}] &\le Ce^{-cL}, \\
        \dP_p [ \Lambda\text{ is not good with respect to occupied sites}]&\le Ce^{-cL}.
    \end{align}
\end{theorem}

To extend the proof of Ref. \cite{dario2023phase} to our scenario, we require the following additional definitions.
\begin{definition}
    Given $\alpha,$ let $C^+, C^- \subseteq \dZ^d$ be a pair of connected clusters of  occupied, respectively vacant, sites of random phase $\alpha$. 
    Then the clusters $C^\pm$ are \textbf{adjacent} if there exists edge $e=ij$ such that $i\in C^+,j\in C^-$ and write $C^+\sim C^-$ if such an edge exists. The set of edges incident to a given pair $C^+, C^-$ of occupied, respectively vacant, clusters is called the \textbf{interface} between $C^+, C^-$ and we denote the set of such edges by {$\{C^+\sim C^-\}$}.
\end{definition}

\begin{definition}
    A finite box $\Lambda \equiv [i]_L \subseteq \dZ^d$ of length $L$ is \textbf{optimal} if the box is good with respect to both occupied and vacant sites of $\alpha$, and the unique, box-spanning clusters $C^+, C^-$ are adjacent.
\end{definition}
Now let $G_{\infty}=\mathbb Z^d$ {or $\tilde{\dZ}^2$} and suppose $p_c^\mathrm{site}(G_\infty) < \min (p,1-p)$. We denote by $C_\infty^+, C_{\infty}^-$ the {a.s. defined} unique infinite clusters of occupied, respectively vacant, vertices in the site percolation.

\begin{lemma}
    \label{thm:optimal}
    If $p_c^\mathrm{site}(G_\infty) < \min (p,1-p)$,
    \begin{equation}
        \dP_p[C_\infty^+ \sim C_\infty^-]=1
    \end{equation}
   and the probability that any given edge $e$ is in the interface of $C^\pm_{\infty}$ is nonzero, i.e.,
    \begin{equation}
        \dP_p[e\in \{C_\infty^+\sim C_\infty^-\}] >0
    \end{equation}
  It follows that 
   \begin{equation}
    \lim_L    \dP_p[\Lambda _L \textit{ is optimal}] = 1.
    \end{equation}
\end{lemma}
\begin{proof}
{Let us restrict ourselves the case where $G_\infty =\dZ^d$; the case where $G_\infty =\tilde{\dZ}^2$ is similar and thus omitted.}
This is a then standard percolation theoretic argument.  By ergodicity, 
 \begin{equation}
        \dP_p[C_\infty^+ \sim C_\infty^-]\in \{0, 1\}.
    \end{equation}
Moreover $\dP_p[C_\infty^+ \sim C_\infty^-]=1$
if and only if, for all $e\in \mathbb Z^d,$ $\dP_p[e\in \{C_\infty^+\sim C_\infty^-\}] >0.$

To decide which, we follow the usual route. By translation invariance, we may assume the edge $e$ in question is the standard unit vector $e_1$, connecting the origin and $(1, 0, \cdots, 0).$  
To show $\dP_p[e_1\in \{C_\infty^+\sim C_\infty^-\}] >0$, we begin by observing that
by continuity of probabilities, if $p_c^\mathrm{site} < \min (p,1-p)$
\begin{equation}
    \lim_{L\to \infty} \dP_p[C^+_\infty, C^-_\infty  \cap \partial \Lambda_L \ne \varnothing]=1.
\end{equation}

Let $S^+_L, S^-_L$ denote the vertices of $\partial \Lambda_L$ which are connected to $\infty$ via occupied, vacant sites in $\Lambda_{L-1}^c$ of random phase $\alpha$, i.e., if $C_x^\pm (\Lambda^c_{L-1})$ denotes the sites connected in $\Lambda^c_{L-1}$ to $x$ via occupied/vacant sites of random phase $\alpha$, then
\begin{equation}
    S_L^\pm = \{x\in \partial \Lambda_L: |C_x^\pm (\Lambda^c_{L-1})|=\infty\}
\end{equation}
Hence,
\begin{equation}
\label{e:Pcond}
    \dP_p[e_1\in \{C_\infty^+\sim C_\infty^-\}]\geq\dE_p [ \dP_p[e_1 \xleftrightarrow[\Lambda_{L-1}]{+}S^+_L, e_1 \xleftrightarrow[\Lambda_{L-1}]{-}S^-_L|S^+_L, S^-_L ]].
\end{equation}
Given nonempty $S_L^\pm$ we can always choose disjoint paths $\gamma^{\pm}$ in $\Lambda_{L-1}$ from the endpoints of $e_1$ to $S_L^\pm$. 
Demanding that the sites of $\gamma^+$ are occupied and the sites of $\gamma^-$ are vacant, we see that there is a constant $c>0$ depending only on $p, L$ so that
\begin{equation}
    \dP_p[e_1 \xleftrightarrow[\Lambda_{L-1}]{+}S^+, e_1 \xleftrightarrow[\Lambda_{L-1}]{-}S^-|S^+_L, S^-_L ]\geq c {1\{S_L^+,S_L^-\ne \varnothing \}}.
\end{equation}
{Furthermore, since the event $\{C^+_\infty, C^-_\infty  \cap \partial \Lambda_L \ne \varnothing\} = \{S_L^+,S_L^-\ne \varnothing\}$ has nonzero probability}, inserting this bound into \eqref{e:Pcond} finishes the proof the the a.s. adjacency of $C^+_{\infty}, C^-_{\infty}$.

To prove the final claim, we now show that if $C^{\pm}_L$ denote the largest internally connected clusters in $\Lambda_L$ then
\begin{equation}
\dP[\Lambda_L \text{ is optimal}]=\dP[\Lambda_L \text{ is good}, C^{+}_L \sim C^{-}_L]=1. 
\end{equation}

Now, by continuity of probabilities  and then Theorem \ref{thm:good-box},
\begin{align}
\label{e:cnty}
1 = \dP_p[C_\infty^+ \sim C_\infty^-] &=\lim_{N} \dP[ C^{+}_\infty, C^{-}_\infty \text{ are adjacent in $\Lambda_N$}]\nonumber \\
 &=\lim_N \dP[ C^{+}_\infty, C^{-}_\infty \text{ are adjacent in $\Lambda_N$, $\Lambda_{2N}$ is good}].
\end{align}
But by definition of a box being good,
\begin{align}
&\{C^{+}_\infty, C^{-}_\infty \text{ are adjacent in $\Lambda_N$, $\Lambda_{2N}$ is good}\} \nonumber\\
&=\{ C^{+}_{2N}, C^{-}_{2N} \text{ are adjacent in $\Lambda_N$ and $\Lambda_{2N}$ is good}\}\cap  \{C^{+}_\infty \cap \Lambda_N,  C^{-}_\infty\cap \Lambda_N \neq \varnothing\}.
\end{align}
Since
\begin{equation}
\dP[C^{+}_\infty\cap \Lambda_N,  C^{-}_\infty\cap \Lambda_N \neq \varnothing] \rightarrow 1 \qquad \text{ as $N \rightarrow \infty,$}
\end{equation}
 we conclude from \eqref{e:cnty} that
\begin{equation}
    \lim _N \dP[C^{+}_{2N}, C^{-}_{2N} \text{ are adjacent in $\Lambda_N$ and $\Lambda_{2N}$ is good}]=1. 
\end{equation}

\end{proof}
\section{$\sH^\mathrm{s}$ is lower bounded by $\sH^\mathrm{c}$}

    Let $L\in \dN_+$ be even and set $n_L=2(2L+1)^d$.  Given a distinct pair $a,b\in (2L+1)\dZ^d,$
    let $e_a, e_b$ denote a fixed pair edges in $\dZ^d$  such that $e_a \in [a]_{L/2}, e_b \in [b]_{L/2}.$
   \newcommand{\emb}{\mathcal S}
    Let
    \begin{equation}
        E_{L,e_a, e_b}  = \{e_v \xleftrightarrow{s} e_v +\partial \Lambda_{L/2}, v\in \{a,b\}, s\in\{+, -\}\}
    \end{equation}
    denote the event that the endpoints of $e_a, e_b$ are connected by both occupied and vacant paths to the boundaries of their respective boxes (arbitrarily choosing to center the boxes at one of the endpoints of $e_v$).
    
    Define the bijective map $\emb=\emb_L:  (2L+1 )\dZ^d \to n_L|\dZ^d$ by $\emb(x)= \frac{n_L}{(2L+1)} x.$ 
    For $\alpha\in E_{L,e_a,e_b }$, define the random environment $r_\alpha \equiv r_{L,\alpha}:n_L|\dZ^d \to \{0,1\}$ via
    \begin{equation}
        (r_\alpha)_{\emb(x)} = 1\{[x]_L \text{ is optimal with respect to } \alpha\}
    \end{equation}

    The heart of the argument for Theorem \eqref{thm:infinite-stability} and \eqref{thm:infinite-quasi-stability} lies in the following theorem.
    \begin{theorem}[$\sH^\mathrm{s}$ is lower bounded by $\sH^\mathrm{c}$]
        \label{thm:lower-bound}
        On the event $E_{L, e_a, e_b},$
        \begin{equation}
            \bra \cos \nabla_{\be_a}\theta \cos \nabla_{\be_b}\theta\ket^\mathrm{s}_{\dZ^d,\alpha,\ve,\beta} 
            \ge \bra \cos \varphi_{\emb (a)} \cos \varphi_{\emb (b)} \ket^\mathrm{c}_{\dZ^d|_{n_L}^{r_{L,\alpha}},\ve,\beta/(2d)} 
        \end{equation}
        where the gradient $\nabla_{\be_v}$ is taken using the orientation $\be_v$ from vacant to occupied vertex. 
    \end{theorem}

\begin{proof}
    \begin{figure}[ht]
        \centering
        \includegraphics[width=0.5\columnwidth]{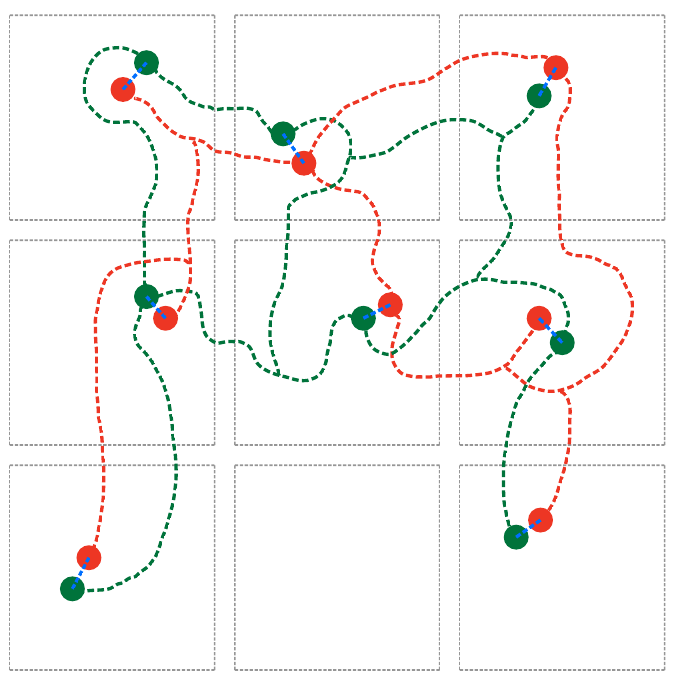}
        \caption{Schematic Transform. The figure denotes the schematic transform on $\tilde{\dZ}^2$. The grey dashed lines denote boxes $[x]_L,x\in (2L+1)\dZ^2$. 
        Boxes with green/red dots are optimal (i.e., the south box is not optimal), where green/red denotes occupied/vacant sites.
        In particular, the green/red dots in each box are adjacent in $\tilde{\dZ}^2$ with an edge denote in blue.
        The dashed green/red lines denote $\gamma_{xy}^\pm$ paths between adjacent boxes.
        In the schematic, the green/red lines seemingly intersect despite denoting distinct types of sites.
        This is due to the non-planarity of $\tilde{\dZ}^2$ and in actuality, they do not intersect at sites.
        }
    \label{fig:schematic-transform}
    \end{figure}
    
    \begin{figure}[ht]
    \centering
    \subfloat[\label{fig:separation-a}]{%
        \centering
        \includegraphics[width=0.3\columnwidth]{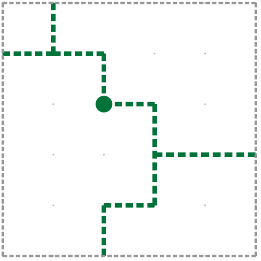}
    }
    \subfloat[\label{fig:separation-b}]{%
        \centering
        \includegraphics[width=0.3\columnwidth]{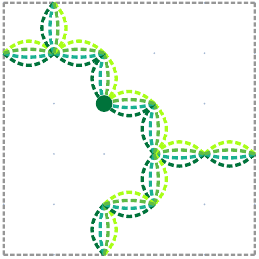}
    }
    \subfloat[\label{fig:separation-c}]{%
        \centering
        \includegraphics[width=0.3\columnwidth]{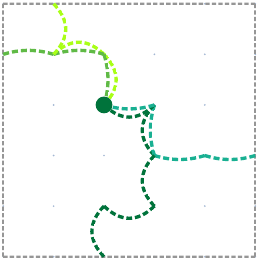}
    }\\
    \subfloat[\label{fig:separation-d}]{%
        \centering
        \includegraphics[width=0.3\columnwidth]{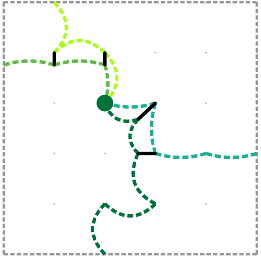}
    }
    \subfloat[\label{fig:separation-e}]{%
        \centering
        \includegraphics[width=0.3\columnwidth]{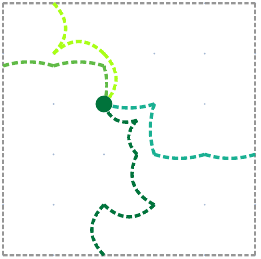}
    }
    \caption{Separation in $\tilde{\dZ}^2$. 
    (a) sketches overlapping paths $\gamma_{xy}^+$ consisting of occupied sites in $[x]_L$, where the coupling along each edge is $\beta$. 
    (b) sketches a multigraph setup equivalent to (a) where the coupling along each edge is $\beta/(2d)$ for $d=2$ in $\tilde{\dZ}^2$. The colors indicate the labeling of edges. 
    (c) sketches a choice of a color/label for each path $\gamma_{xy}^+$ so that the resulting paths are edge-disjoint.
    (d) sketches a multigraph setup equivalent to (c) where vertices are duplicated and the coupling along edges between duplicated vertices (black lines) is $\infty$.
    (e) sets the coupling in (d) along edges between duplicated vertices (black lines) to $0$.
    }
    \label{fig:separation}
    \end{figure}
    The construction is similar to that given in Sec. (5.2.2) of Ref. \cite{dario2023phase}, so we will be brief.

    Fix $\alpha$ and let $x\in (2L+1)\dZ^d$ be given such that  $[x]_L$ is optimal with respect to $\alpha$.  
    Let $C_x^+, C_x^-$ denote the unique crossing clusters within $[x]_L$ consisting of occupied, respectively vacant, sites.
    By definition of optimality, $C_x^+\sim C_x^-$ and we can find a pair of vertices $x^\pm \in C_x^\pm$ so that $x^+ x^-$ is an interfacial edge between $C_x^+, C_x^-$.  Choose a distinguished pair $x^+, x_-$ for each optimal box. If $E_{L, e_a, e_b}$ occurs, then for  $x\in  \{a, b\}$, choose $x^\pm =a^\pm,$ resp. $b^\pm.$ 
    
    Now if  $[x]_L,[y]_L$ are optimal, the fact that the boxes are good implies that the unique crossing clusters $C_x^\pm, C_y^\pm$ within of $[x]_L,[y]_L$ must be connected via occupied, resp. vacant, sites within $[x]_L\cup[y]_L$.
    Hence, there exist (simple) paths $\gamma_{xy}^\pm$ consisting of occupied/vacant sites in $[x]_L\cup [y]_L$ between $x^\pm$ and $y^\pm$.          The lengths of the paths $\gamma_{xy}^\pm$ are bounded above by $n_L$.

    Let $\sG_\infty$ be the infinite subgraph of $\dZ^d$ consisting of vertices and edges belonging to the paths $\gamma_{xy}^\pm,$ where $x, y$ ranges over all possible adjacent optimal boxes $[x]_L,[y]_L$. 
    Let $\sI_\infty$ denote the collection of the distinguished edges $x^+ x^-$ in $[x]_L$ for $x\in (2L+1)\dZ^d$ such that $[x]_L$ is optimal, (see Fig. \ref{fig:schematic-transform}). 
    In relation to Definition \eqref{def:H-clean}, $\sG_\infty$ is of the type considered there with bipartition $G^+_{\infty}, G^-_{\infty}$ defined by the subgraphs of $C^+_{\infty}, C^-_{\infty}$ respectively and ``interlayer'' edges $ \sI_\infty$ is a clean model.  
    Thus the Ginibre inequality, Theorem \eqref{T:Gin}, implies that
    \begin{equation}
    \label{e:Gin1}
        \bra \cos \nabla_{\be_a}\theta \cos \nabla_{\be_b}\theta\ket^\mathrm{s}_{\dZ^d,\alpha,\ve,\beta} 
        \ge \bra \cos \nabla_{\be_a}\theta \cos \nabla_{\be_b}\theta\ket^\mathrm{s}_{\sG_\infty,\alpha,\ve,\beta}
        = \bra \cos \varphi_{\be_a} \cos \varphi_{\be_b}\ket^\mathrm{c}_{\sG_\infty,\ve,\beta} .
    \end{equation}
    
    Note that $\sG_{\infty}$ is not the bilayer model on $\dZ^d|_{n_L}$ or even $\dZ^d|_{n_L}^r$ for some choice of $r$ due to the percolative imperfections in the choice of $\gamma_{xy}^\pm.$  The remainder of the proof uses Ginibre's inequality to further mold the $\sG_\infty$ into the relatively regular $\dZ^d|_{n_L}^{r_{L,\alpha}}$.   
    
    There are a few sources of irregularity to deal with.  First, for $[y]_L,[z]_L$ adjacent to $[x]_L$, the pair of paths $\gamma_{xy}^+$ and $\gamma_{xz}^+$ may overlap away from $x^{+}$ (and similarly for $G^-$).  
    Given $x$, consider the set of $y$ such that $[y]_L, [x]_L$ are adjacent and optimal. 
    There are at most $2d$ such adjacent boxes. Moreover, as $y$ varies, the edges of $\gamma_{xy}^+$ can only overlap within $[x]_L.$ 
    This observation allows us to modify the graphs within each box independently.
    
    To resolve overlaps within a fixed box, one can perform the operations described in Sec. (5.2.2) of Ref. \cite{dario2023phase} to separate the paths.  These operations are summarized in Fig. \ref{fig:separation}.  
    First, for each edge of $\mathcal G_\infty$ within $[x]_L,$ we may split the coupling strength of the edge, $\beta,$ as $2d \times \frac{\beta}{2d}.$ 
    We then replace $\sG_\infty \cap [x]_L$ with a weighted multi-graph in which each edge $e\in \sG_\infty \cap [x]_L$ with weight $\beta$ gets replaced by $2d$ edges $e_1, \cdots, e_{2d},$ each having edge weight $\beta/2d.$  
    Let $\sM_\infty$ denote the resulting multi-graph.  
    We now choose, arbitrary paths $\tilde{\gamma}_{xy}^{\pm}$ in $\sM_\infty$ which project to $\gamma_{xy}^{\pm}$ in $\sG_\infty$ and are edge disjoint in $\sM_\infty$.
    Relabeling as necessary, we may assume each path $\tilde{\gamma}_{xy}^{\pm}$ uses edges having the same index.
    
    To render the paths $\tilde{\gamma}_{xy}^{\pm}$ \textit{vertex} disjoint, 
    replace each vertex $v\in \sM_\infty\backslash \{x^+, x^-: x\in n_L|\mathbb Z^d\}$  with $2d$ copies $v_1, \cdots v_{2d}$ and add vertical edges between $v_i, v_{i+1},$ having weight $\infty.$ We regard $\tilde{\gamma}_{xy}^{\pm}$ as passing through vertices with index $i$ if its edges use the index $i$ copies.  
    Let $\sG_\infty^0$ denote the induced subgraph obtained by retaining only the vertices $\{x^+, x^-: x\in n_L|\mathbb Z^d\},$ the corresponding interfacial edges $x^+x^-,$
    and the indexed edges and vertices corresponding to  $\tilde{\gamma}_{xy}^{\pm}.$  Finally reduce weight of any remaining vertical edges to $0.$  
    
    Note that all edges of $\mathcal G_\infty^0$ have weight $\beta/2d.$  
    Moreover, for every pair of adjacent optimal boxes $[x]_L,[y]_L,$ the vertices $x^{\pm}$ and $y^{\pm}$ are connected by copies of $\tilde{\gamma}^{\pm}_{xy}$ which are mutually edge disjoint, as well as mutually vertex disjoint except at the endpoint $x^{\pm}$ and $y^{\pm}.$
    By the Ginibre inequality, it then follows once again that
    \begin{equation}
    \label{e:Gin2}
        \bra \cos \varphi_{\be_a} \cos \varphi_{\be_b} \ket^\mathrm{c}_{\sG_\infty,\ve,\beta} 
        \ge \bra \cos \varphi_{\be_a} \cos \varphi_{\be_b} \ket^\mathrm{c}_{\sG_\infty^0,\ve,\beta/2d}.
    \end{equation}
    \begin{figure}
    \centering
    \subfloat[\label{fig:duplication-a}]{%
        \centering
        \includegraphics[width=0.3\columnwidth]{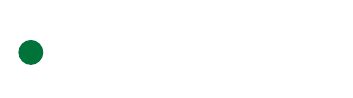}
    }
    \subfloat[\label{fig:duplication-b}]{%
        \centering
        \includegraphics[width=0.3\columnwidth]{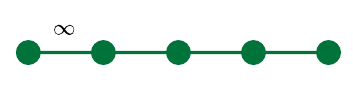}
    }
    \subfloat[\label{fig:duplication-c}]{%
        \centering
        \includegraphics[width=0.3\columnwidth]{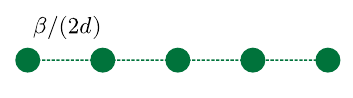}
    }
    \caption{Duplication. (a) sketches a single site possibly interacting with other sites (not drawn). (b) sketches a setup equivalent to (a) where the horiztonal edges have $\infty$ coupling. (c) sketches the lower bound (in the sense of Ginibre) of (c) by setting the edge coupling to $\beta/(2d)$.
    }
    \label{fig:duplication}
    \end{figure}
    Next we deal with irregularity in the lengths of $\tilde{\gamma}^{\pm}_{xy}.$ Note that the path separation procedure does not change the length of each path and thus separated paths have length $\le n_L$. By duplicating spins and lowering the coupling from $\infty$ to $\beta/(2d)$ as above (and as sketched in Fig. \ref{fig:duplication}), we can add vertices in series as necessary so as to assume that all separated paths have length exactly $n_L$ and all edges have weight $\beta/(2d)$.  
    Denote the resulting graph by $\sG_\infty^1$.

    Finally, observe that the Gibbs measure $\mu_{\sG_\infty^1,\ve,\beta/2d}^\mathrm{c}$ is the marginal distribution of the clean model $\mu_{\dZ|_{n_L}^{r_{L,\alpha}},\ve,\beta/2d}^\mathrm{c}$: 
    the vertices on ``dangling end paths" corresponding to optimal boxes $[x]_L$ and neighboring non-optimal boxes $[y]_L$ as well as pairs of neighboring non-optimal boxes having been integrated out.
    After this final identification, we have created the clean Hamiltonian on $\dZ^d|_{n_L}^{r_{L,\alpha}}$ with inverse temperatures $\beta/(2d)$, interlayer interaction $1/h$ and random environment $r_{L,\alpha}$.  
    By the Ginibre inequality,  
    \begin{equation}
    \label{e:Gin3}
        \bra \cos \varphi_{\be_a} \cos \varphi_{\be_b} \ket^\mathrm{c}_{\sG_\infty^0,\ve,\beta/2d}\ge
        \bra \cos \varphi_{\be_a} \cos \varphi_{\be_b} \ket^\mathrm{c}_{\sG_\infty^1,\ve,\beta/2d} 
        = \bra \cos \varphi_{\emb (a)} \cos \varphi_{\emb(b) }\ket^\mathrm{c}_{\dZ|_{n_L}^{r_{L,\alpha}},\ve,\beta/2d}.
    \end{equation}
    The statement of the Theorem follows by combining \eqref{e:Gin1}, \eqref{e:Gin2}, \eqref{e:Gin3}. 
   
\end{proof}

\begin{remark}
    It's worth mentioning that another route to controlling $\sH^\mathrm{s}$ would be to bound the Hamiltonian below by setting $1/h=0.$  This comparison falls within those for which Ginibre's inequality applies.   Then, we could apply Theorem \eqref{thm:dario-main} directly to argue for LRO/QLRO on the infinite clusters $C^{\pm}_{\infty}.$ In a second pass we would then argue that the term
    \begin{equation}
        -\frac{1}{2\ve}\sum_{x\in V} (\nabla_x^\alpha \cdot \cos \nabla\theta)^2,
    \end{equation}
    creates an effective Ising interaction on the effective two-spin space  defined by the ordering directions on $C^+_{\infty}, C^-_{\infty}.$
    We have chosen to rerun the Dario-Garban construction  more fully since 
    \begin{enumerate}
        \item The explicit construction elucidates the relation between strongly disordered $\sH^\mathrm{s}$ and weakly interacting $\sH^\mathrm{c}$ beyond the heuristic symmetry arguments discussed in Sec. \eqref{sec:strong-disorder}.
        \item The construction motivates the study of the weakly interacting clean model $\sH^\mathrm{c}$ (in two dimensions in particular).  
        In Sec. \eqref{sec:H-clean}, we take up the study of $\sH^\mathrm{c}$ in order to provide further insight into the nature  of ordering for the strongly disordered model $\sH^\mathrm{s}$ in $d=2$ dimensions.
    \end{enumerate}
\end{remark}
Note that the inequality in the previous theorem depends on the random phase $\alpha$ through $r_{L,\alpha}$.
To remove this dependence, we follow  Ref. \cite{dario2023phase} and collect the following lemmas.
\begin{lemma}[\cite{liggett1997domination}]
    \label{lem:stoch-dom}
    For any $d,k\ge 1$ and $p\in (0,1)$, there exists $\hat{p}(d,k,p)<1$ such that the following holds: for any subset $\Lambda \subseteq \dZ^d$, if $\{X_i,i\in \Lambda\}$ is a \textbf{$k$-dependent field} of Bernoulli variables (i.e., the $\sigma$-algebras generated by $\{X_i,i\in A\}$ and $\{X_i,i\in B\}$ are independent for any subsets $A,B\subseteq \Lambda$ with graph distance $\mathrm{dist}(A,B) >k$) satisfying $\dP[X_i=1]\ge \hat{p}$ for all $i\in \Lambda$, then $\{X_i,i\in \Lambda\}$ stochastically dominates an i.i.d. field of Bernoulli variables on $\Lambda$ with parameter $p$.
\end{lemma}
\begin{lemma}[Proposition (3.1), (4.1), (5.1) of Ref. \cite{dario2023phase}]
    \label{lem:well}
    Consider $\dZ^d|_n$ with $d\ge 2,n>1$. Let $p_0\equiv p_0(d,J,\ve) = (1+e^{-4dJ-J/\ve})^{-1}$ and let $(r_x)_{x\in n|\dZ^d}$ be an i.i.d. Bernoulli field with parameter $p_0$.
    Then for $n>1$, 
    \begin{equation}
        \dE_{p_0(d,J,h)}\bra \cos \varphi_a \cos \varphi_b\ket^\mathrm{c}_{\dZ^d|_n^{r},\ve,J,J'} \ge \bra \cos \varphi_a \cos \varphi_b\ket^\mathrm{c}_{\dZ^d|_n,\ve,J/2,J'}
    \end{equation}
\end{lemma}
\begin{proof}
    The proof follows almost exactly in Propositions (3.1), (4.1), (5.1) of Ref. \cite{dario2023phase}.
    The only difference is the bound on the ratio of partition functions. 
    Given a site $x\in n|\dZ^d$, let $r^{(0)}$ be the random environment on $n|\dZ^d$ obtained from $r$ which is $0$ at $x$, and is $r_y$ for sites $y\ne x$.
    Similarly define $r^{(1)}$ to be the random environment on $n|\dZ^d$ which $=1$ on site $x$ and is $r$ otherwise.
    Then 
    \begin{align}
    \label{e:compare}
        \frac{\sZ_{\Lambda_{Nn}|_n^{r^{(0)}},\ve,J, J'}^\mathrm{c}}{\sZ_{\Lambda_{Nn}|_n^{r^{(1)}},\ve, J, J'}^\mathrm{c}} & = \left\bra \exp\left(-J \sum_{y\in \Lambda_{Nn}|n: y\sim x} \sum_{\ell=\pm} \cos\nabla_e \vartheta^\ell -\frac{J}{\ve} \cos^2 \varphi_x\right)\right\ket^\mathrm{c}_{\Lambda_{Nn}|_n^{r^{(1)}}, \ve, J, J}\\
        &\ge e^{-4J d -J/\ve}
    \end{align}
    The $\ve$ dependence in the lower bound gives rise to the $\ve$ dependence in $p_0$.
\end{proof}

\begin{remark}
    Note that $p_0(d, J,h)$ is increasing with respect to $J$.  It is this effect which necessitates the use of multiple coupling strengths in Definition \ref{def:H-clean-multiT}.  
    Crucially, this trick  decouples the choice of $L$ (as defined in Theorem \eqref{thm:lower-bound}) from the inverse temperature $\beta$.
    Otherwise, if $J=J'\sim \beta,$ then to have a high enough density of optimal boxes, $L$ must be chosen large depending on $\beta.$  But the larger $L$ is the LESS ferromagnetic the RHS of \eqref{e:compare} effectively is and the argument does not close.
   
\end{remark}
\begin{proof}[Proof of Theorem \eqref{thm:infinite-stability} and \eqref{thm:infinite-quasi-stability}]
    Let us first consider the case of $\dZ^d,d\ge 3$.
    Let $x,y\in \dZ^d$  and let $x',y'$ denote arbitrary vertices neighboring $x,y$, respectively.
    For any given $L$, let $a,b\in \dZ^d$ be such that $a-b=x-y$ and choose $e_a\in [a]_{L/2}$ and $e_b\in [b]_{L/2}$ for some $a,b\in (2L+1)\dZ^d$ 
     so that $e_a =x-x'$ and $e_b=y-y'$.
    Consider the event $E_{L,e_a,e_b}$ as defined in Theorem \eqref{thm:lower-bound}.
    Provided $\|x-y\|\geq L$,
    \begin{align}
        \dP_p[ E_{L,e_a,e_b}] &= \dP_p[e_a \xleftrightarrow{s} e_a +\partial \Lambda_{L/2},s\in\{+, -\}]^2 \\
        &\ge \dP_p[e_a \in \{C_\infty^+ \sim C_\infty^-\}]^2
    \end{align}
    By Theorem \eqref{thm:optimal}, the RHS is nonzero and thus the event has probability bounded away from $0$ uniformly in $\|x-y\|\geq L$.
    
    Let $p_0(d,\beta,\ve)$ be as in Lemma \eqref{lem:well} and $J_0^\mathrm{c}(d)$ be as in Proposition \eqref{prop:H-clean}.
    Let $\hat{p} <1$ be that determined by $p_0(d,2J_0^\mathrm{c}(d),\ve_0)$ via Lemma \eqref{lem:stoch-dom}.
    Define the field of random variables $r_\alpha \equiv r_{L,\alpha}: n_L|\dZ^d \mapsto \{0,1\}$ as in Theorem \eqref{thm:lower-bound} and note that $r_\alpha$ is a 1-dependent field of Bernoulli random variables.
    By Theorem \eqref{thm:optimal}, we see that $\dP_p [(r_\alpha)_x = 1]$ can be made arbitrarily close to 1 provided that $L$ is sufficiently large.
    
    Henceforth, let $L=L(d,p,\ve_0)$ be sufficiently large depending on $d,p,\ve_0$ (and $\norm{x-y}$ be sufficiently large depending on $L$ and thus on $d,p,\ve_0$) so that
    \begin{equation}
        \dP_p[ (r_\alpha)_z=1|E_{L,e_a,e_b}] \ge \hat{p}, \quad \forall z\in \dZ^d.
    \end{equation}
    Then the conditional distribution of random field $r_\alpha$ with respect to event $E_{L,e_a,e_b}$ stochastically dominates the Bernoulli random field $r$ of probability $p_0(d,2 \beta_0^\mathrm{c}(d),\ve_0)$ on $(2L+1)\dZ^d \cong n|\dZ^d$ where $n\equiv n_L$ {(note that we have written $r_\alpha$ and $r$ to differentiate the two)}.

    Note that
    \begin{align}
        \bra \tau_x \tau_y\ket^\mathrm{s}_{\dZ^d,\alpha,\ve,\beta} 
        &= \frac{1}{\ve^2} \sum_{e_x \sim x, e_y\sim y} 1\{\nabla_{\be_x} \alpha,\nabla_{\be_y} \alpha\ne 0\}\bra \cos \nabla_{\be_x}\theta \cos \nabla_{\be_y}\theta\ket^\mathrm{s}_{\dZ^d,\alpha,\ve,\beta}
    \end{align}
    where the RHS sum is over edges $e_x, e_y$ containing $x, y$ respectively and both edges $\be_x,\be_y$ are oriented from vacant to occupied vertices.
    By translation invariance and then applying Ginibre's inequality,
    \begin{align}
        \dE_p \left\bra \prod_{i=x,y} 1\{\nabla_{\be_i} \alpha\ne 0\} \cos \nabla_{\be_i}\theta \right\ket^\mathrm{s}_{\dZ^d,\alpha,\ve,\beta} &= \dE_p \left\bra \prod_{i=a,b}1\{\nabla_{\be_i} \alpha\ne 0\} \cos \nabla_{\be_i}\theta \right\ket^\mathrm{s}_{\dZ^d,\alpha,\ve,\beta}\\
        &\ge \dE_p\left[\left\bra \prod_{i=a,b} \cos \nabla_{\be_i}\theta\right\ket^\mathrm{s}_{\dZ^d,\alpha,\ve,\beta}\Bigg| E_{L, e_a, e_b}\right]\dP_p [E_{L,e_a,e_b}],
    \end{align}
    where again, the edges are oriented from vacant to occupied vertices.
    
    Let
    \begin{equation}
        \beta_0^\mathrm{s}(d,p,\ve_0)   = 2d \max\left(2 J_0^\mathrm{c}(d), J_0^{\mathrm{c}\prime}(d,n)\right)
    \end{equation}
    where $n=n_L$ depends on $d,p,\ve_0$.
    By Theorem \eqref{thm:lower-bound}, for $\alpha\in E_{L, e_a, e_b}$ we have
    \begin{align}
        \bra \cos \nabla_{\be_a}\theta \cos \nabla_{\be_b}\theta\ket^\mathrm{s}_{\dZ^d,\alpha,\ve,\beta} &\geq  \bra \cos \varphi_{\emb (a)} \cos \varphi_{\emb (b)}\ket^\mathrm{c}_{\dZ^d|_{n}^{r_{\alpha}} ,\ve, \beta/2d} .
    \end{align}
    Where $\emb = \emb_L$ depends on $L$ (and thus on $d,p,h_0)$.
    From this inequality it follows that 
    \begin{align}
            \dE_p [\bra \cos \nabla_{\be_a}\theta \cos \nabla_{\be_b}\theta\ket^\mathrm{s}_{\dZ^d,\alpha,\ve,\beta} |E_{L,e_a,e_b}]
            &\ge \dE_p [\bra \cos \varphi_{\emb (a)} \cos \varphi_{\emb (b)}\ket^\mathrm{c}_{\dZ^d|_{n}^{r_{\alpha}},\ve,\beta/2d}|E_{L,e_a,e_b}]
    \end{align}
    By Lemma \eqref{lem:stoch-dom}, the RHS of this inequality is bounded by
    \begin{align}
        \dE_p [\bra \cos \varphi_{\emb (a)} \cos \varphi_{\emb (b)}\ket^\mathrm{c}_{\dZ^d|_{n}^{r_{\alpha}},\ve,\beta/2d}|E_{L,e_a,e_b}]
        \label{eq:stoch-dom-renorm-applied}
        &\ge \dE_{p_0(d,2J_0^\mathrm{c}(d),\ve_0)}\bra \cos \varphi_{\emb (a)} \cos \varphi_{\emb (b)}\ket^\mathrm{c}_{\dZ^d|_{n}^{r},\ve,\beta/2d} \\
        \label{eq:stoch-dom-p-applied}
        &\ge \dE_{p_0(d,2J_0^\mathrm{c}(d),\ve)}\bra \cos \varphi_{\emb (a)} \cos \varphi_{\emb (b)}\ket^\mathrm{c}_{\dZ^d|_{n}^{r},\ve,\beta/2d} 
    \end{align}
    where the second  inequality \eqref{eq:stoch-dom-p-applied} utilizes $p_0(d,\beta,\ve) \le p_0(d,\beta,\ve_0).$ 
    Using the Ginibre inequality and then Lemma \eqref{lem:well}, we see that for $\beta > \beta_0^\mathrm{s}(d,p,h_0)$,
    \begin{align}
        \dE_{p_0(d,2J_0^\mathrm{c}(d),\ve)}\bra \cos \varphi_{\emb (a)} \cos \varphi_{\emb (b)}\ket^\mathrm{c}_{\dZ^d|_{n}^{r},\ve,\beta/2d}        \label{eq:monotone-beta-applied}
        &\ge \dE_{p_0(d,2J_0^\mathrm{c}(d),\ve)}\bra \cos \varphi_{\emb (a)} \cos \varphi_{\emb (b)}\ket^\mathrm{c}_{\dZ^d|_{n}^{r}, \ve, 2J_0^\mathrm{c}(d),J_0^{c\prime}(d,n)} \\ 
        \label{eq:well-inequality-apply}
        &\ge \bra \cos \varphi_{\emb (a)} \cos \varphi_{\emb (b)}\ket^\mathrm{c}_{\dZ^d|_{n},J_0^{\mathrm{c}}(d), J_0^{\mathrm{c}\prime}(d,n),\ve}
    \end{align}
    Finally, by Proposition \eqref{prop:H-clean}, the proof of Theorem \eqref{thm:infinite-stability} is completed.
    The proof of Theorem \eqref{thm:infinite-quasi-stability} is similar and omitted.
\end{proof}

\section{$\dZ_2$ degeneracy of $\sH^\mathrm{c}$}
\label{sec:H-clean}
We have seen that the effective Hamiltonian $\sH^\mathrm{s}$ in the strong disorder limit $\ve\gg1$ is lower bounded by a variant of the clean Hamiltonian $\sH^\mathrm{c}$ in the weak inter-layer interaction $1/\ve\ll 1$.
Motivated by this  fact, in this section we consider $\sH^\mathrm{c}$ in more detail, motivated by the assumption that the pair of models have broadly similar behavior, especially in the weak limit $1/\ve \to 0$ on $\dZ^d$.

In fact, to make things easier still, we further assume, discussed in Sec. \eqref{sec:weak-disorder}, that the average phase and phase difference  degrees of freedom decouple in the weak limit $1/\ve \to 0$. Since the variable of interest at present is the phase difference, $\varphi$, the model reduces to the Hamiltonian
\begin{equation}
    \label{eq:H-clean-weak}
    \sH^\mathrm{cw}_{G,\ve}(\varphi) = -\sum_{e\in E} \cos \nabla_e \varphi -\frac{1}{2\ve} \sum_{i\in V} \cos^2 \varphi_i
\end{equation}
This is an $XY$ model with an additional weak double well potential with minima at $0, \pi.$  Intuitively, one would expect $\Z_2$ ordering at these minimia, but the situation in two dimensions seems to be delicate, as we point out below.

Note that  the expected ordering of $\sH^s$ is in the $Y$ direction, i.e., $e^{i\phi}\sim^\mathrm{s} \pm i$.  Through Eq. \eqref{eq:heuristic}, this corresponds in $\sH^\mathrm{cw}$ to a phase difference $\varphi$ ordering in the $X$ direction. 


\subsection{Absence of Ordering in $Y$}
\label{sec:absense-of-ordering-in-Y}
For completeness, we first observe that $\sH^\mathrm{cw}$ does not order along he $Y$ direction for any finite temperature in any dimension.
The proof is a simple modification of the infrared bound and thus will be mostly left in the appendix.

\begin{prop}[Absence of ordering in $Y$]
    \label{cor:absence-Y}
    Let $\dZ_L^d$ denote the torus with linear length $L$. Then
    \begin{equation}
        \left\bra \left[\frac{1}{|\dZ_L^d|} \sum_{x\in \dZ_L^d} \sin \varphi_x \right]^2\right\ket^\mathrm{cw}_{\dZ_L^d,\ve,\beta} \le \frac{\ve}{\beta |\dZ_L^d|} \to 0, \quad L\to \infty
    \end{equation}
\end{prop}
\begin{proof}
    By the infrared bound in Corollary \eqref{cor:infrared}, take $\hat{n}$ to be uniformly equal to $\hat{y}$ (the unit vector in the $Y$ direction) and thus the statement follows.
\end{proof}

\subsection{Dimension $d=2$}

By Proposition \eqref{cor:absence-Y}, the effective Hamiltonian $\sH^\mathrm{cw}$ in Eq. \eqref{eq:H-clean-weak} does not order in the $Y$ direction regardless of temperature $T$ and interaction strength $1/\ve >0$. 
On the other hand, it satisfies Ginibre's inequality \cite{ginibre1970general}, and we may conclude that if $d\ge 3$, the model orders in the $X$ direction and exhibit LRO below a temperature independent of interaction strength, i.e., $T^\mathrm{cw}_\mathrm{LRO} \sim 1$ in the limit $1/\ve \to 0$.

However, when the spatial dimension $d=2$ dimension, due to the underlying BKT phase of the XY model, it's less clear whether such a statement holds.  Indeed, 
by Ginibre and Proposition \eqref{cor:absence-Y}, the effective Hamiltonian $\sH^\mathrm{cw}$ has at least quasi-LRO in the $X$ component below some $T_\mathrm{QLRO}^\mathrm{cw}\sim 1$ \textit{independent} of interaction strength $1/\ve \to 0$.  Whether the system exhibits LRO with $T_\mathrm{LRO}^\mathrm{cw}\sim 1$ remains an open question.  Conventional wisdom (i.e. relevance of the $\cos(2\varphi)$ interaction term under renormalization group transformations \cite{Seibergetal2021}) suggests it should order.  Our next result shows that if so, the ordering is extremely weak for temperatures of order one.

 Let $\bra \cdots \ket_{\partial,\Lambda_N,\ve,\beta}^\mathrm{cw}$ denote the Gibbs measure with respect to Hamiltonian $\sH^\mathrm{cw}$ in Eq. \eqref{eq:H-clean-weak} on $\Lambda_N$ with wired boundary conditions, i.e., spins $\phi_i=0$ are fixed at $i\in \partial \Lambda_N$, and let $\bra \cdots \ket_{\partial,\dZ^2,\ve,\beta}^\mathrm{cw}$ denote the thermodynamic limit, again  existing due to Ginibre's inequality.
\begin{theorem}[Modified McBryan-Spencer]
    \label{thm:mcbryan}
    For every sufficiently small $\delta >0$, there exists $\beta_0,\ve _0 > 0$ and constants $C_0,C_1>0$ such that if $\beta > \beta_0$ and $\ve > \ve_0$, 
    \begin{equation}
        \bra \cos \varphi_0\ket_{\partial,\dZ^2,\ve,\beta}^\mathrm{cw} = \limsup_{N \to \infty}\bra \cos \varphi_0\ket_{\partial,\Lambda_N,\ve,\beta}^\mathrm{cw} \le C_1 \exp\left[ -\frac{1}{1+\delta}\frac{1}{8\pi\beta} \log \ve\right]
    \end{equation}
   
\end{theorem}
\begin{proof}
    For notation simplicity, we shall omit the subscript $\Lambda_N$.
    Following McBryan-Spencer \cite{mcbryan1977decay}, we apply an imaginary shift $\varphi \mapsto \varphi +iu$ for some $u:\Lambda_N \to \dR$ with $u_x =0$ for all $x\in \partial\Lambda_N.$  Defining the corresponding energy shift by 
    \begin{align}
        -\delta_u \sH^{cw}(\varphi)&=- \sH^{cw}(\varphi+iu)+ \sH^{cw}(\varphi) =\sum_{e\in \Lambda_N} (\cosh \nabla_e u -1) \cos \nabla_e \varphi +\frac{\beta}{4\ve} \sum_{i\in \Lambda_N} (\cosh (2u_i) -1)\cos(2\varphi_i) \nonumber\\
        &-i\sum_{e\in \Lambda_N} \sinh \nabla_e u \sin \nabla_e \varphi -i \frac{1}{2\ve} \sum_{i\in \Lambda_N} \sinh (2u_i) \sin(2\varphi_i)
    \end{align}
we have
    \begin{align}
        \bra \cos\varphi_0\ket_{\partial,\ve,\beta}^\mathrm{cw} &= \frac{1}{\sZ_{\partial,\ve,\beta}^\mathrm{cw}}\int d\varphi e^{i\varphi_0} e^{-\beta \sH_{\partial,\ve}^\mathrm{cw}(\varphi)} \\
        &=\frac{1}{\sZ_{\partial,\ve,\beta}^\mathrm{cw}}e^{-u_0}\int d\varphi e^{i\varphi_0}  e^{-\beta \sH_{\partial,\ve}^\mathrm{cw}(\varphi)}   e^{-\delta_u \sH^{cw}(\varphi)}
        \end{align}
Using the triangle inequality and the bound
\begin{equation}
    |e^{-\beta \delta_u \sH^{cw}(\varphi)}|\leq e^{-\beta \Re( \delta_u \sH^{cw}(\varphi))},
\end{equation}
it follows that
        \begin{align}
        |\bra \cos\varphi_0 \ket_{\partial,\ve,\beta}^\mathrm{cw}| 
        &\le \exp\left[ \beta \left(-\frac{u_0}{\beta} +\sum_{e\in\Lambda_N} (\cosh \nabla_e u-1) +\frac{1}{4\ve} \sum_{i\in\Lambda_N} (\cosh (2u_i)-1)\right)\right].
    \end{align}

    The remainder of the proof seeks to optimize the upper bound by a clever choice for $u$.  In the usual application, the term proportional to $1/h$ is absent and the bound is (nearly) optimized by the harmonic function which is $1$ at $0$ and $0$  on $\partial \Lambda_N.$  In the present case the presence of the second term means that this trial function needs to be revisited.  The natural guess is to take $u$ to be the corresponding massive Green function, with $m^2=1/h,$ the only question being whether the higher order nonlinear terms allow the bound to be optimized further.  As we now argue, these terms are insignificant.

   Since it may not be completely standard,  let us first run through a sketch of the proof before introducing the technical details.
    Rewrite $\sE(u)$ as 
    \begin{equation}
        \label{eq:functional}
        \sE(u)\equiv \underbrace{-\frac{u_0}{\beta} +\frac{1}{2} u \left(-\Delta_N +\frac{1}{\ve}\right)u }_{\text{at most quadratic in }u} +\frac{1}{4\ve}\sum_{i\in \Lambda_N} (\cosh (2u_i) -2u_i^2-1).
    \end{equation}
   Note that the mass term $m^2 = 1/\ve$ is chosen so that $\cosh(2u)-2u^2 -1 \sim u^4$ is small when $u$ is small.
    Optimize the first two terms by taking
    \begin{equation}
        u= \frac{1}{\beta}\left(-\Delta_N +\frac{1}{\ve}\right)^{-1} \delta_0
    \end{equation}
    Heuristically take $N\to \infty$ so that the we can replace the solution with the $d=2$ dimensional continuum massive Green functions, i.e.,
    \begin{equation}
        u_x \sim  \frac{1}{\beta} C_x, \quad C_x = \frac{1}{2\pi} K_0\left(\frac{\norm{x}}{\sqrt{\ve}}\right)
    \end{equation}
    Where $K_0$ is the modified Bessel function of the second kind.
    Hence, by properties of the massive Green function, we note that $|\nabla_e u| \lesssim 1/\beta$ uniformly for edges $e$ and thus for sufficiently large $\beta$ (independent of $\ve$), the lattice gradients can be made sufficiently small.
    Note that $K_0(\omega)$ has the following asymptotics,
    \begin{equation}
        K_0(\omega) \sim 
        \begin{dcases} 
            -\log\left(\frac{\omega}{2}\right) + \gamma, & \omega \to 0 \\
            \sqrt{\frac{\pi}{2\omega}} e^{-\omega}, & \omega \to \infty
        \end{dcases}
    \end{equation}
    where $\gamma$ is the Euler–Mascheroni constant.
The first two terms in Eq. \eqref{eq:functional} scale as 
    \begin{equation}
        -\frac{1}{2\beta} C_0 \approx -\frac{1}{8\pi \beta} \log \ve, \quad R\to \infty
    \end{equation}
    On the other hand,
    passing from a lattice sum to a continuum integral, we can partition the higher order contribution as
    \begin{equation}
        \sum_{i} (\cosh(2u_i) -2u_i^2 -1) \sim \underbrace{\int_{r \lesssim \sqrt{\ve}} \exp\left(-\frac{1}{\pi\beta}\log\left(\frac{r}{2\sqrt{\ve}}\right)\right)rdr}_{\sim \ve} + \underbrace{\int_{r\gtrsim \sqrt{\ve}} \frac{\ve}{r^2} e^{-4r/\sqrt{\ve}},rdr}_{\sim\ve}
    \end{equation}
    where we used the fact that $\cosh (2u)-1-2u^2 \sim e^{2|u|}$ for large $u$ and $\sim u^4$ for small $u$.
    Thus the higher order term in Eq. \eqref{eq:functional} is bounded as $\ve \to \infty$ and the theorem statement follows.

    We now make the analysis precise.
    For simplicity, we shall use $A,B,A_1,B_1,...>0$ to denote constants that may change values from line to line.
    Fix $\delta >0$ and write $1/m^2 = \ve(1+\delta)$. 
    Let $G,G_N$ denote the Green function on $\dZ^2,\Lambda_N$ as given in Proposition \eqref{prop:prob-Green} and \eqref{prop:prob-Green-N}.
    Then define
    \begin{equation}
        u_x = \frac{1}{1+\delta}\frac{1}{\beta} G_N(x)
    \end{equation}
    By definition, we see that $u_x=0$ on the boundary $x\in \partial\Lambda_N$.
    By Corollary \eqref{cor:gradient}, note that there exists constant $A$ (independent of $N$) such that
    \begin{equation}
        |\nabla_e u| \le \frac{A}{\beta}, \quad e\in \Lambda_N
    \end{equation}
    Hence, for sufficiently large $\beta$, we see that 
    \begin{equation}
        \cosh(\nabla_e u)-1 \le \frac{1}{2}(1+\delta) (\nabla_e u)^2
    \end{equation}
    In particular, we see that 
    \begin{align}
        -\frac{u_0}{\beta} + \sum_{e\in \Lambda_N} (\cosh (\nabla_e u) -1) +\frac{1}{2 \ve} u^2 &\le -\frac{u_0}{\beta} + \frac{1}{2} (1+\delta) u(-\Delta +m^2) u\\
        &\le -\frac{1}{1+\delta} \frac{1}{2\beta^2} G_N(0) \\
        \limsup_{N\to \infty} \left[ -\frac{u_0}{\beta} + \sum_{e\in \Lambda_N} (\cosh (\nabla_e u) -1) +\frac{1}{2 \ve} u^2\right] &\le -\frac{1}{1+\delta} \frac{1}{2\beta^2} G(0) \\
        &\le -\frac{1}{1+\delta} \frac{1}{8\beta^2} \log\ve +\frac{A}{\beta^2}
    \end{align}
    To bound the error term, note that
    \begin{equation}
        \cosh(2u) -1-2u^2 \le A \min\left(e^{2|u|}, u^4\right)
    \end{equation}
    From Proposition \eqref{prop:small-mx} and \eqref{prop:large-mx}, we see that if $1/\sqrt{\ve}$ is bounded above, then there exists constants $C,C_1,C_2,D_1,D_2 >0$ independent of $\ve$ such that 
    \begin{equation}
        G(x) \le 
        \begin{dcases} 
            -C_1 \log\left(\frac{1}{m\norm{x}}\right) + C_2, & m\norm{x}  \le C \\
            \frac{D_1}{\sqrt{m\norm{x}}} e^{-D_2 m\norm{x}}, & m\norm{x} \ge C
        \end{dcases}
    \end{equation}
    Hence,
    \begin{align}
        \sum_{x\in \Lambda_N} (\cosh (2u_x) -1-2u_x^2) &\le A \left[\exp\left(\frac{A}{\beta} \log m^{-1}\right)\right.\\
        &\quad\quad\quad+\sum_{x\in \dZ^2,x\ne 0} \exp\left(\frac{A}{\beta} \log\left(\frac{1}{m\norm{x}}\right)\right) 1\{ m\norm{x}\le C\} \nonumber\\
        &\quad\quad\quad\left.+ \sum_{x\in \dZ^2} \frac{1}{(m\norm{x})^2} \exp\left(-B m\norm{x}\right) 1\{ m\norm{x}\ge C\}\right] \nonumber\\
        &\le A \left[\exp\left(\frac{A}{\beta} \log m^{-1}\right)\right.\\
        &\quad\quad\quad+\int_{\dR^2} dx \exp\left(\frac{A}{\beta} \log\left(\frac{1}{m\norm{x}}\right)\right) 1\{ m\norm{x}\le C\} \nonumber\\
        &\quad\quad\quad\left.+ \int_{\dR^2} \frac{dx}{(m\norm{x})^2} \exp\left(-B m\norm{x}\right) 1\{ m\norm{x}\ge C\} \right] \nonumber\\
        &\le A \left[\frac{1}{m^{A/\beta}} +\frac{1}{m^2} \int_0^C \rho^{1-A/\beta} d\rho +\frac{1}{m^2}\int_C^\infty \frac{d\rho}{\rho} e^{-B\rho}\right] \\
        &= O\left(\frac{1}{m^2}\right) = O\left(\ve\right)
    \end{align}
    Where the first inequality uses the fact that $0\le G_N(x)\le G(x)$ (which follows immediately from Proposition \eqref{prop:prob-Green} and \eqref{prop:prob-Green-N}), second inequality uses the fact that the integrand is decreasing with respect to $\norm{x}$ and the last equality uses the fact that $\beta$ is sufficiently large. 
    Therefore, we see that the error term is bounded, and thus the statement follows.
\end{proof}

%% file: app.tex
\section{Reflection Positivity and the Infrared Bound}

In this section, we shall prove a modification of the infrared bound for the clean Hamiltonian $\sH^\mathrm{cw}$ in Eq. \eqref{eq:H-clean-weak}.
In particular, we require the following lemmas.
\begin{lemma}
    \label{lem:conic}
    Let $p(\tau)$ denote a polynomial over Ising variables $\tau \in \{\pm 1\}^S$ for some finite subset $S$, with nonnegative coefficients.
    Then 
    \begin{equation}
        \sum_{\tau \in \{\pm 1\}^S} p(\tau) \ge 0
    \end{equation}
\end{lemma}
\begin{proof}
    Since the coefficients of the polynomial are nonnegative, it's sufficient to show that 
    \begin{equation}
        \sum_{\tau \in \{\pm 1\}^S} \prod_{i\in S} \tau_i^{n_i} \ge 0
    \end{equation}
    For all $n_i\in \dN$.
    Indeed, this is straightforward by noticing that the summation is zero if there exists an odd integer $n_i$.
\end{proof}
\begin{lemma}
    \label{lem:XY-to-Ising}
    Let $f(e^{i\theta}):\dS^1 \to \dR$ denote a continuous function. Then 
    \begin{equation}
        \int_{-\pi}^\pi d\theta f(\theta) = \frac{1}{4} \int_{-\pi}^\pi d\theta \sum_{\xi,\eta =\pm 1} f(\xi|\cos\theta| +i\eta|\sin \theta|)
    \end{equation}
\end{lemma}
\begin{proof}
    The proof follows straightforwardly if we denote $\xi,\eta$ as the sign of the $X,Y$ components of $e^{i\theta}$.
\end{proof}
\begin{theorem}[Reflection Positivity and Gaussian Domination]
    \label{thm:gaussian-dom}
    Define the following Hamiltonian on the torus $\dZ^d_L$ of length $L$,
    \begin{align}
        \label{eq:H-clean-weak-torus}
        \sH^\mathrm{cw}_{\dZ^d_L, \ve,s}(\sigma) &= \frac{1}{2} (\sigma+s)\left(-\Delta + \ve^{-1}\dnum^Y\right)(\sigma+s)\\
        &\equiv  \frac{1}{2} \sum_{i\in \dZ^d_L}(\sigma_i+s_i) \left[\left(-\Delta + \ve^{-1}\dnum^Y\right)(\sigma+s)\right]_i
    \end{align}
    where $\sigma = e^{i\varphi}$ denotes the spin in $\dS^1$, $-\Delta$ is the discrete lattice Laplacian, and $\dnum^Y$ is the projection operator of the spin $\sigma$ onto its $Y$ axis, i.e., $\dnum^Y \sigma_i = \sin \varphi_i$. If $a >0$ amd $\hat{n}:\dZ_L^d \to \dS^1$, then the partition functions satisfy
    \begin{equation}
        \sZ^\mathrm{cw}_{\dZ_L^d, \ve, s=a\hat{n}} \le \sZ^\mathrm{cw}_{\dZ_L^d, \ve, s=a\hat{x}}
    \end{equation}
    Where $\hat{x}$ is uniformly the unit vector in the $X$-direction in $\dS^1$.
\end{theorem}
Note that up to unimportant additive constants, the Hamiltonian in Eq. \eqref{eq:H-clean-weak} and Eq. \eqref{eq:H-clean-weak-torus} with $h=0$ are equal and thus they define the same finite-volume Gibbs measure.
\begin{proof}
    For notation simplicity, we shall omit the subscripts $\dZ_L^d,\ve$.
    Note that for $h=a\hat{n}$,
    \begin{equation}
        \sH^\mathrm{cw}_{s}= \frac{1}{2} \sum_{e\in \dZ_L^d}  |\nabla_e (\sigma+s)|^2 +\frac{1}{2\ve} \sum_i (\sigma^Y_i +s_i^Y)^2
    \end{equation}
    where $\sigma^Y,s^Y$ denote the $Y$-component of $\sigma,s$, respectively.
    Note that the second term acts onsite and thus $\sH^\mathrm{cw}_{s}$ is reflection positive \cite{friedli2017statistical}. 
    More specifically, consider a vertical cut $\gamma$ of $\dZ_L^d$ through some edge so that the vertices of $\dZ_L^d$ are partitioned into $V^\pm$ (with induced edges $E^\pm$), and let $E^0$ denote the edge intersecting the cut. 
    Then
    \begin{align}
        \sH^\mathrm{cw}_{s}&= \frac{1}{2} \sum_{e\in E^\pm} |\nabla_e (\sigma+s)|^2 \\
        &\quad+ \frac{1}{2} \sum_{e=i^\pm \in E^0:x^\pm \in V^\pm} \left[|\sigma_{x^+}+s_{x^+}|^2 +|\sigma_{x^-}+s_{x^-}|^2 -2 (\sigma_{x^+}+s_{x^+})\cdot (\sigma_{x^-}+s_{x^-})\right]\\
        &\quad +\frac{1}{2\ve} \sum_{x\in V^\pm} (\sigma_x^Y+s_x^Y)^2 \\
        &= A+ \gamma B -2\sum_{x^+ \in V^+: x^+\sim V^-}C_{x^+} \bar{C}_{x^+}
    \end{align}
    where $\sim$ denotes adjacency and 
    \begin{align}
        A &= \frac{1}{2} \sum_{e\in E^+} |\nabla_e (\sigma+s)|^2+ \frac{1}{2} \sum_{x^+ \in V^+:x^+\sim V^-}  |\sigma_{x^+}+s_{x^+}|^2+\frac{1}{2\ve} \sum_{x\in V^+} (\sigma_x^Y+s_x^Y)^2 \\
        \gamma B &= \frac{1}{2} \sum_{e\in E^-} |\nabla_e (\sigma+s)|^2+ \frac{1}{2} \sum_{x^- \in V^-:x^-\sim V^+}  |\sigma_{x^-}+s_{x^-}|^2+\frac{1}{2\ve} \sum_{x\in V^-} (\sigma_x^Y+s_x^Y)^2 \\
        C_{x} &= \sigma_{x}+s_{x}
    \end{align}
    Here, we write $\gamma B$ to be the reflection of $B$ with respect to the vertical cut $\gamma$, i.e., $(\gamma B)_x = B_{\gamma x}$ where $\gamma x$ is the site obtained from site $x$ after reflection across the vertical cut $\gamma$.
    Since $A,B$ and $C_{x^+}, x^+\in V^+$ are locally supported on $V^+$, we see by reflection positivity,
    \begin{equation}
        (\sZ^\mathrm{cw}_{s})^2 \le \sZ^\mathrm{cw}_{s^+}\sZ^\mathrm{cw}_{s^-}
    \end{equation}
    where $s^+: \dZ_L^d \to \dS^1$ is such that $s^+=s$ on $V^+$ and $=\gamma s$ on $V^-$, and similarly for $s^-$.

    Let $\sN(\hat{n})$ denote the collection of edges $e=xy$ such that $\hat{n}_x \ne \hat{n}_y$ and let $\hat{n}^\star:\dZ_L^d \to \dS^1$ be a configuration which maximizes the partition $\hat{n}\mapsto \sZ^\mathrm{cw}_{s=a\hat{n}}$ and has the minimum $|\sN(\hat{n}^\star)|$.
    Suppose that $|\sN(\hat{n}^\star)|>0$. 
    Then there exists a vertical cut through one of the edges in $\sN(\hat{n}^\star)$ and thus 
    \begin{equation}
        \sN(\hat{n}^{\star+}) +\sN(\hat{n}^{\star-}) <2\sN(\hat{n}^{\star})
    \end{equation}
    Without loss of generality, we can assume that $\sN(\hat{n}^{\star+}) \le \sN(\hat{n}^{\star-})$ so $\sN(\hat{n}^{\star+}) < \sN(\hat{n}^{\star})$.
    Since $\hat{n}^\star$ maximizes the partition function $\hat{n}\mapsto \sZ^\mathrm{cw}_{a\hat{n}}$, the previous inequality obtained from reflection positivity implies that
    \begin{equation}
        \sZ^\mathrm{cw}_{a\hat{n}^\star} \le \sZ^\mathrm{cw}_{a\hat{n}^{\star+}}
    \end{equation}
    And thus we reach a contradiction. 
    Hence, $\sN(\hat{n}^\star) =\emptyset$ and thus $\hat{n}^\star$ uniformly points in some direction for all sites in $\dZ_L^d$.

    We further claim that $\hat{n}^\star =\hat{x}$ uniformly maximizes the partition function. 
    Indeed, let $\hat{n} = e^{i\gamma}$. Then we see that 
    \begin{equation}
        \frac{\partial}{\partial \gamma} \sZ_{a\hat{n}}^\mathrm{cw} = -\frac{\beta a}{\ve} \sum_{i \in \dZ_L^d} \int d\varphi e^{-\beta \sH^\mathrm{cw}_{a\hat{n}}(\varphi)} \cos \gamma (\sin \varphi_i + a \sin \gamma)
    \end{equation}
    Since $\sH^\mathrm{cw}_{a\hat{n}}$ depends on $\hat{n}$ only implicitly via its $Y$-component, we see that it's sufficient to show that if $\gamma \in (0,\pi/2)$, then 
    \begin{equation}
        \int d\varphi e^{-\beta \sH^\mathrm{cw}_{a\hat{n}}(\varphi)}  \sin \varphi_i \ge 0
    \end{equation}
    Indeed, note that (using $\propto$ to denote proportional up to a nonnegative constant)
    \begin{align}
        \int d\varphi e^{-\beta \sH^\mathrm{cw}_{a\hat{n}}(\varphi)}  \sin \varphi_x &\propto \int d\varphi \sin \varphi_x \exp\left[-\frac{\beta}{2\ve} \sum_{y\in \dZ_L^d} (\sin^2 \varphi_y + 2a \sin \varphi_y \sin \gamma ) \right] \exp\left[\frac{\beta}{2} \sigma (-\Delta\sigma) \right]\\
        &\propto \int d\varphi \sin \varphi_x \exp\left[\frac{\beta a}{\ve} \sin \gamma \sum_{y} \sin \varphi_y  \right]\exp\left[\frac{\beta \ve}{4} \sum_y \cos (2\varphi_y)\right]\exp\left[\frac{\beta}{2} \sigma (-\Delta\sigma) \right]
    \end{align}
    Where the last equality follow from the transform $\varphi\mapsto \pi-\varphi$. 
    It's then sufficient to show that 
    \begin{equation}
        \int d\varphi \prod_x \sin^{n_x} \varphi_x \exp\left[\frac{\beta \ve}{4} \sum_y \cos (2\varphi_y)\right]\exp\left[\frac{\beta}{2} \sigma (-\Delta\sigma) \right] \ge 0, \quad \forall n_i \in \dN
    \end{equation}
    Note that 
    \begin{equation}
        \exp\left[\frac{\beta}{2} \sigma (-\Delta\sigma) \right] \propto \exp \left[ \beta \sum_{e=ij} (\xi_i\xi_j |X_i X_j| +\eta_i\eta_j|Y_i Y_j|)\right] = p(\xi_i\xi_j,\eta_i\eta_j: e=ij)
    \end{equation}
    Where $X_i,Y_i$ denote the $X,Y$ components of spin $\sigma$ with signs $\xi_i,\eta_i =\pm 1$, and $p$ denotes a polynomial of arguments $\xi_i\xi_j,\eta_i\eta_j$ with nonnegative coefficients.
    Hence, by Lemma \eqref{lem:conic} and \eqref{lem:XY-to-Ising}, we have
    \begin{equation}
        \int d\varphi \prod_x \sin^{n_x} \varphi_x \exp\left[\frac{\beta \ve}{4} \sum_y \cos (2\varphi_y)\right]\exp\left[\frac{\beta}{2} \sigma (-\Delta\sigma) \right] \propto \int d\varphi_i \sum_{\xi,\eta}  \prod_{i} \eta_i^{n_i} p(\xi_i\xi_j,\eta_i\eta_j) \ge 0
    \end{equation}
\end{proof}

\begin{corollary}[Infrared Bound]
    \label{cor:infrared}
    Let $\hat{n}:\dZ^d_L \to \dS^1$ and $s=a\hat{n}$. Then
    \begin{equation}
        \bra \exp\left[ -\beta \sigma(-\Delta +\ve^{-1} \dnum^Y)\sigma\right]\ket^\mathrm{cw}_{\dZ_L^d,\ve,\beta} \le \exp\left[\frac{\beta}{2} s(-\Delta +\ve^{-1} \dnum^Y)h\right]
    \end{equation}
    Where $\bra \cdots\ket^\mathrm{cw}_{\dZ_L^d,\ve,\beta}$ is with respect to the Hamiltonian $\sH^\mathrm{cw}_{\dZ_L^d,\ve}$.
    In particular,
    \begin{equation}
        \bra [\sigma(-\Delta +\ve^{-1}\dnum^Y) \hat{n}]^2\ket^\mathrm{cw}_{\dZ_L^d,\ve,\beta} \le \frac{1}{\beta} \hat{n} (-\Delta+\ve^{-1} \dnum^Y) \hat{n}
    \end{equation}
\end{corollary}

\begin{proof}
    For notation simplicity, we shall omit the subscripts $\dZ_L^d,\ve$.
    Note that 
    \begin{equation}
        \sZ^\mathrm{cw}_{a\hat{x}} = \int d\varphi e^{-\beta \sH^\mathrm{cw}(\varphi)}
    \end{equation}
    And that for any $\hat{n}:\dZ_L^d \to \dS^1$, 
    \begin{equation}
        \sZ_{a\hat{n}}^\mathrm{cw} = \exp\left[-\frac{\beta}{2} a^2 \hat{n} (-\Delta +\ve^{-1} \dnum^Y) \hat{n}\right] \int d\varphi  e^{-\beta \sH^\mathrm{cw}(\varphi)} \exp\left[ -\beta a \sigma (-\Delta +\ve^{-1}\dnum^Y)\hat{n}\right] \\
    \end{equation}
    Hence, by the previous Gaussian domination bound in Theorem \eqref{thm:gaussian-dom}, we see that 
    \begin{equation}
        \bra \exp\left[ -\beta \sigma(-\Delta +\ve^{-1} \dnum^Y)\sigma\right]\ket^\mathrm{cw}_{\dZ_L^d,\ve,\beta} \le \exp\left[\frac{\beta}{2} h(-\Delta +\ve^{-1} \dnum^Y)h\right]
    \end{equation}
    In particular, if we take $a\to 0$, then Taylor expansion implies that 
    \begin{equation}
        \bra [\sigma(-\Delta +\ve^{-1}\dnum^Y) \hat{n}]^2\ket^\mathrm{cw}_{\dZ_L^d,\ve,\beta} \le \frac{1}{\beta} \hat{n} (-\Delta+\ve^{-1} \dnum^Y) \hat{n}
    \end{equation}
\end{proof}

\section{Lattice Green Function}
\begin{definition}
    \label{def:lattice-Green}
    The lattice Green function with mass $m$ on $\dZ^d,d\ge 2$ is defined as 
    \begin{equation}
        G(x) \equiv \int_{(-\pi,\pi)^d} \frac{d k}{(2\pi)^d} \frac{e^{ikx}}{\tilde{k}^2/2d +m^2}, \quad \tilde{k}^2 =\sum_{i=1}^2 \left(2 \sin\left(\frac{k_i}{2}\right)\right)^2
    \end{equation}
    Where we have written it in this manner to closely mirror the continuum Green function, i.e., $\tilde{k}^2 \approx k^2$ in the limit $k\to 0$.
\end{definition}

\begin{remark}
    Using the fact that $\lambda^{-1} = \int_0^\infty e^{-s\lambda}ds$, the lattice Green function can be rewritten as 
    \begin{equation}
        G(x) = \int_0^\infty e^{-s m^2} \prod_{i=1}^2 \bar{I}_{x_i} \left(\frac{s}{2}\right), \quad \bar{I}_n (s) = e^{-s} I_{n} \left(s\right)
    \end{equation}
    Where $I_n(s)$ is the modified Bessel function of the first kind and has the following integral form 
    \begin{equation}
        I_n(s) = \int_{-\pi}^\pi \frac{d\theta}{2\pi} e^{s \cos \theta + in\theta}
    \end{equation}
\end{remark}

\begin{lemma}[Lemma B.1 of Ref. \cite{michta2021asymptotic}]
    There exists $A,\delta>0$ such that
    \begin{equation}
        \bar{I}_\nu (\nu^2 s) \le A \left(\frac{1}{\nu\sqrt{s}}e^{-\delta^2 /s}1\{2\nu s\ge 1\}+e^{-\delta \nu} 1\{2\nu s <1\} \right)
    \end{equation}
\end{lemma}
\begin{proof}
    For simplicity, we will use $A >0$ as a constant which may change values from line to line.
    Moreover, we can assume without loss of generality that $\nu >0$.
    Note that in Lemma B.1 of Ref. \cite{michta2021asymptotic}, the inequality was proven for sufficiently large $\nu$. 
    The reason is that Ref. \cite{michta2021asymptotic} utilized the fact that 
    \begin{equation}
        \bar{I}_\nu (\nu t) = L_\nu (t) (1+o(1)), \quad \nu \to \infty
    \end{equation}
    Where the convergence is uniform in $t$ and 
    \begin{equation}
        L_\nu (t) = \frac{1}{\sqrt{2\pi \nu}} \frac{e^{\nu \psi(t)}}{(1+t^2)^{1/4}}, \quad \psi(t) = -t +\sqrt{1+t^2} -\sinh^{-1} (t^{-1})
    \end{equation}
    For sufficiently large $\nu$, we then have the upper bound $\bar{I}_\nu (\nu t) \le 2 L_\nu (t)$.
    One can relax this condition by utilizing Eq. (B.12) in Ref. \cite{frohlich1981kosterlitz}, in which it was shown that 
    \begin{equation}
        \bar{I}_\nu (\nu t) =L_\nu (t) \left(1+O\left(\frac{1}{n(1+t^2)}\right)\right)
    \end{equation}
    Hence, if $n\ge 1$, then $\bar{I}_\nu (\nu t) \le A L_\nu (t)$ and the remainder of the proof is exactly the same as that shown in Lemma B.1 of Ref. \cite{michta2021asymptotic}.
\end{proof}

\begin{prop}
    \label{prop:small-mx}
    Let $G(x)$ denote the lattice Green function on $\dZ^2$ with mass $m$. For any sufficiently large constant $C>0$, there exists constants $C_1,C_2 >0$ (independent of $m$) such that if $0<\norm{x}_2 a \le C$ with $x\in \dZ^2$, then
    \begin{equation}
        G(x) \le C_1 \log \left(\frac{1}{m \norm{x}_2}\right) +C_2
    \end{equation}
\end{prop}
\begin{proof}
    We will use $A,B,\delta >0$ as constants which may change values from line to line.
    Also note that all norms on $\dZ^2 \subseteq \dR^2$ are equivalent (the ratio of distinct norms are bounded above and below uniformly in $x$) and thus we will always use $\norm{x} \equiv \norm{x}_2$ as the default norm.
    Without loss of generality, we shall assume that $x_1 \ge x_2 \ge 0$ and that $x\ne 0$ (since at $x=0$, it's know that the inequality holds, e.g., Proposition 8.32 of Ref. \cite{friedli2017statistical}).
    Using the previous lemma, we see that if $x_i > 0$, then 
    \begin{align}
        \bar{I}_{x_i} \left(\frac{s}{2} \right) &= \bar{I}_{x_i}\left( x_i^2 \frac{s}{2x_i^2}\right) \\
        &\le A \left(\frac{1}{\sqrt{s}} e^{-(\delta  x_i)^2/s} 1\{ s \ge x_i\} +e^{-\delta x_i} 1\{ s < x_i\} \right)
    \end{align}
    Alternatively, if $x_i =0$, we note that
    \begin{equation}
        \bar{I}_{x_i} \left(\frac{s}{2} \right) \le A \min\left(1,\frac{1}{\sqrt{s}}\right)
    \end{equation}
    We then claim that
    \begin{equation}
        G(x) \le A \int_0^\infty e^{-s m^2} \left[e^{-\delta^2 \norm{x} }  1\{s\le B\norm{x}\} + \frac{1}{s} e^{-\delta^2 \norm{x}^2 /s}\right]
    \end{equation}
    Indeed, if $x_2 >0$, then
    \begin{equation}
        \prod_{i=1}^2 \bar{I}_{x_i} \left(\frac{s}{2}\right) \le A \left[
            \frac{1}{s} e^{-(\delta \norm{x}_2)^2/s} 1\left\{s \ge x_1\right\} + 
            e^{-\delta x_1} \frac{1}{\sqrt{s}} e^{-(\delta x_2)^2/s} 1\left\{x_1 > s \ge x_2\right\} +
            e^{-\delta \norm{x}_1} 1\left\{x_2 >s \right\}
            \right]
    \end{equation}
    Note that for $s\ge x_2$, we have 
    \begin{equation}
        \frac{1}{\sqrt{s}} e^{-(\delta x_2)^2/s} \le \frac{1}{\sqrt{x_2}} \le A e^{-\delta x_2}
    \end{equation}
    where we utilized the fact that $x_2 \ge 1$. 
    Hence, for $x_i \ne 0$ for both $i=1,2$, we have
    \begin{align}
        \prod_{i=1}^2 \bar{I}_{x_i} \left(\frac{s}{2}\right) \le A \left[
            \frac{1}{s} e^{-(\delta \norm{x}_2)^2/s} 1\left\{s \ge \norm{x}_\infty\right\} 
            +e^{-\delta \norm{x}_1} 1\left\{s < \norm{x}_\infty\right\}\right]
    \end{align}
    The case is similar if $x_2=0$.
    Utilizing the fact that all norms are equivalent, the claim then follows.

    We can then bound the first term as follows 
    \begin{align}
        \int_0^\infty e^{-s m^2} e^{-\delta \norm{x} }  1\{s\le B\norm{x}\} &= e^{-\delta \norm{x}} \frac{1}{m^2} \left(1- e^{-m^2 B\norm{x}}\right) \\
        &\le B\norm{x} e^{-\delta \norm{x}} \\
        &\le A
    \end{align}
    The second term is bounded by 
    \begin{align}
        \int_0^\infty e^{-s m^2}  e^{-\delta^2 \norm{x}^2/s} \frac{ds}{s} &= \int_{-\infty}^\infty e^{-2\delta m\norm{x} \cosh t} dt \\
        &= K_0 (2\delta m\norm{x} ) \\
        &\le C_1 \log\left(\frac{1}{m\norm{x}}\right)+A
    \end{align}
    Where $K_0$ is the modified Bessel function of the second kind.
    Hence, the statement follows.
\end{proof}

\begin{prop}[Theorem 1.3 in Ref. \cite{michta2021asymptotic}]
    \label{prop:large-mx}
    Let $G(x)$ be the lattice Green function on $\dZ^2$ with mass $m$. 
    Let $m$ be upper bounded (say $m \le 1$). Then there exists constants $C,C_1,C_2 >0$ (dependent on upper bound, but independent of $m$) such that if $\norm{x}_2 m \ge C$, then 
    \begin{equation}
        G(x) \le \frac{C_1}{\sqrt{m\norm{x}_2}} \exp\left( -C_2 m\norm{x}_2 \right)
    \end{equation}
\end{prop}

\section{Probabilistic Interpretation}
\label{app:prob-G}
As discussed in detail in Ref. \cite{friedli2017statistical} (also see Ref. \cite{lawler2012intersections} for the massless lattice Green function), the lattice Green function with mass $m$ has a nice probabilistic interpretation which can be naturally restricted to finite boxes $\Lambda_N \subseteq \dZ^d$.
More specifically, let $S_n$ denote a simple symmetric random walk on $\dZ^d$ with probability $e^{-\lambda^2} = (1+m^2)^{-1}$ to be killed at any step and denote the corresponding probability as $\bE$. 
Then the lattice Green function $G(x)$ in Definition \eqref{def:lattice-Green} can be rewritten as 
\begin{prop}[Theorem 8.26 in Ref. \cite{friedli2017statistical}]
    \label{prop:prob-Green}
    Let $\Delta$ be the Laplacian operator on $\dZ^d,d\ge 2$. 
    Let $x,y\in \dZ^d$ and define
    \begin{equation}
        G(x,y) \equiv  e^{-\lambda^2} \bE_x \left[\sum_{n=0}^\infty 1\{S_n =y \}\right] = e^{-\lambda^2} \sum_{n=0}^\infty e^{-\lambda^2 n}\bP_x [S_n =y]
    \end{equation}
    Then $G(x) =G(0,x)$ and
    \begin{equation}
        G = \left(-\frac{\Delta}{2d} + m^2 \right)^{-1}
    \end{equation}
\end{prop}

Let $\Lambda_N$ denote a finite box of $\dZ^d$ with linear length $N$ and let $\Delta_N = \dnum_N \Delta \dnum_N$ denote the projected Laplacian, i.e., $\dnum_N$ is the projection operator onto $\Lambda_{N-1} $ so that $(\dnum_N u)_i =u_i 1\{ i\in \Lambda_{N-1}\}$.
If we wish to consider the inverse of the $(-\Delta_N/2d +m^2)$ restricted to $\Lambda_N$, then the probabilistic interpretation provides the following simple representation.
\begin{prop}[Theorem 8.26 in Ref. \cite{friedli2017statistical}]
    \label{prop:prob-Green-N}
    Let $d\ge 2$ and $\Lambda_N\subseteq\dZ^d$ and $x,y\in \dZ^d$ and define 
    \begin{equation}
        G_N(x,y) \equiv e^{-\lambda^2} \bE_x \left[\sum_{n=0}^{\tau -1} 1\{S_n=y\}\right]
    \end{equation}
    Where $\tau$ is the stopping time when $S_n$ leaves $\Lambda_{N-1}$, i.e., $\tau = \inf \{n\ge 0: S_n \notin  \Lambda_{N-1}\}$.
    Then
    \begin{equation}
        G_N = \left(-\frac{\Delta_N}{2d} + m^2 \right)^{-1}
    \end{equation}
    when considered as operators restricted to functions on $\Lambda_{N-1}$, i.e., $\dC^{\Lambda_{N-1}}$.
\end{prop}

It's then clear that $0\le G_N(x,y) \le G(x,y)$ and that $G_N(x,y) \nearrow G(x,y)$ as $N \to \infty$.
Note that $G_N(x,y)=0$ if either $x$ or $y$ is not in $\Lambda_{N-1}$.
By reversing paths, it's also clear that $G_N(x,y) = G_N(y,x)$.
Let us also further collect the following result, which follows the same proof as in Ref. \cite{lawler2012intersections}.
\begin{prop}[Modified version of Proposition 1.5.8 in Ref. \cite{lawler2012intersections}]
    \label{prop:finite-infinite-connection}
    Let $G_N,G$ be defined as in Proposition \eqref{prop:prob-Green-N} and \eqref{prop:prob-Green}. Then
    \begin{equation}
        G_N(x,y) = G(x,y) -\sum_{z\in \partial\Lambda_N} \dP_x[S_\tau = z] G(z,y)
    \end{equation}
    where $\tau$ the stopping time that $S_n$ leaves $\Lambda_{N-1}$.
\end{prop}
\begin{proof}
    Ref. \cite{lawler2012intersections} considers the scenario of a symmetric random walk (corresponding to massless Green functions). However, the proof is exactly the same for the massive lattice Green function.
\end{proof}
\begin{corollary}
    \label{cor:gradient}
    Write $G_N(x) \equiv G_N(0,x)$ where $G_N$ is defined in Proposition \eqref{prop:prob-Green-N}.
    Then there exists constant $C>0$ such that for any edge $e\in \dZ^d$,
    \begin{equation}
        |\nabla_{\be} G_N| \le C
    \end{equation}
\end{corollary}
\begin{proof}
    Note that by Proposition \eqref{prop:finite-infinite-connection}, it's sufficient to show that $|\nabla_{\be} G| \le C$ for some constant $C>0$ for all edges $e$.
    For notation simplicity, we shall use $A,B>0$ to denote constants that may change from line to line.
    Without loss of generality, let $e=xy$ where $y_1=x_1+1$ and $y_i=x_i$ for all $i\ge 2$.
    Note that
    \begin{align}
        |\nabla_{\be} G| &\le \int_{(-\pi,\pi)^d} \frac{dk}{(2\pi)^2} \frac{1-\cos k_1}{\tilde{k}^2/2d +m^2} \\
        &\le A\int_{(-\pi,\pi)^d} dk \frac{1}{\norm{k}}\\
        &\le A
    \end{align}
    where use the fact that $1/\norm{k}$ is locally integrable for $d\ge 2$
\end{proof}